\documentclass[11pt]{article}
\usepackage{mydef2col}
\usepackage{markArticle}

\newcommand{\Lap}[1]{\widehat{#1}}
\newcommand{\Xfdt}{\mathcal{X}}
\newcommand{\Kmat}{\mathbf{K}}

\title{Beyond Rough Volatility: Decoupling Memory and Scaling via a Generalized Langevin Equation}

\shorttitle{Decoupling Memory and Scaling via GLE}

\author{
    \authorstyle{ Andrey Itkin}
    \newline \newline
    \institution{FRE department, Tandon School of Engineering, New York University, email: \url{aitkin@nyu.edu}} \\
}

\date{\today}

\begin{document}

\maketitle

\lettrineabstract{Borrowed from non-equilibrium statistical mechanics, the generalized Langevin equation (GLE) is imported as a framework for stochastic volatility to address the structural limitations of fractional Brownian motion (fBm), the standard engine of rough volatility. The fBm forces a single parameter to set two logically independent properties at once: how volatility scales and how it remembers. The GLE separates them using a memory kernel $K$, a potential $U$, and a noise covariance $C$. Memory becomes a measurable object, and an asymmetric potential supplies a lever on variance skew that the price-variance correlation cannot reach. Physical-measure tests on public datasets decisively reject two constrained corners of the class, a memoryless leverage effect and time-reversal symmetry, while the central rough scaling constraint is left identification-limited rather than refuted. The paper reports these limits honestly, and validation on industry-grade data remains a valuable direction. The risk-neutral construction and the joint SPX--VIX calibration will be developed in a companion paper.
}

\section{Introduction} \label{intro}

Rough volatility is a statement about memory. Throughout, $S_t$ denotes the price
of the underlying at time $t$, $V_t$ its instantaneous variance, and
$\sigma_t=\sqrt{V_t}$ the instantaneous volatility. In the rough specification
$\log\sigma_t$ is modelled as a fractional Brownian motion with Hurst index
$H\approx0.1$, a value far below the $H=1/2$ of ordinary Brownian motion. The
resulting dynamics are non-Markovian, and they reproduce two features of the
data: the observed smoothness of realized volatility time series,
\cite{GJR2018}, and the steepening of the at-the-money implied skew as option
maturity shortens, \cite{ALV2007, Fukasawa2017}. The construction is further
supported by a microstructural limit theorem. Nearly unstable Hawkes processes
with heavy-tailed excitation converge, after rescaling, to the rough Heston model,
\cite{ElEuchFukasawaRosenbaum2018, HorstXu2024}, and imposing absence of
arbitrage on the market impact function forces that excitation to decay as a
power law, \cite{JusselinRosenbaum2020}.

We start from an observation about this construction which is easy to state and,
we believe, has not been drawn out. Fractional Brownian motion carries a single
parameter, $H$, and that parameter is asked to specify two logically independent
properties. The first is how the process scales, meaning how its fluctuations
grow when time is stretched. The second is how rough it is locally, meaning how
its increments behave over vanishing intervals. Memory is a property of the
second kind. A model built on fractional Brownian motion therefore cannot
strengthen the memory while holding the scaling fixed, or the reverse.

\cref{sec:corner} makes this precise, and doing so requires one change of
coordinate that we describe here in words. A process that is self-similar with
index $H$ is not stationary, so it has no memory kernel in the usual sense.
Replacing calendar time $t$ by logarithmic time $\tau=\log t$ and rescaling the
process by $t^{-H}$ produces a stationary process, and a stationary process does
have a well defined memory kernel. We write that kernel $\mathcal{K}(\tau)$,
reserving the plain symbol $K$ for kernels in calendar time. In these
coordinates a fractional driver has $\mathcal{K}(\tau)\sim\tau^{-2H}$ as
$\tau\to0$, while its scaling index remains $H$. The memory exponent is $2H$ and
the scaling exponent is $H$. One number sets both, which is the constraint we
propose to remove.

The generalized Langevin equation (GLE) removes it. Let $Y_t$ be a real-valued
latent process, with no direct market interpretation, which drives the variance
through a link function $\varphi$ according to $V_t=\varphi(Y_t)$. Taking
$\varphi=\exp$ enforces positivity of the variance automatically. Let $Y_t$ obey
\begin{equation}  \label{eq:intro-gle}
m\ddot Y_t = -\int_0^tK(t-s)\dot Y_s\,d s - U'(Y_t) + \xi_t ,
\end{equation}
where a dot denotes differentiation with respect to time. The three objects on
the right are the content of the model. The scalar $m\geq0$ multiplies the second
derivative and is an inertia, with $m=0$ the overdamped case in which the
variance driver carries no momentum of its own. The function $K$ is the memory
kernel, also called the friction or dissipation kernel. It is convolved against
the past velocity $\dot Y_s$ for $s<t$, so it measures how strongly the earlier
history of the variance resists its present rate of change. A kernel
concentrated at the origin gives an ordinary diffusion with no memory, while a
kernel with a heavy tail retains influence from the distant past. The function
$U$ is a potential, and $-U'(Y_t)$ is the restoring force pulling the variance
driver back toward the minimum of $U$, which plays the role of mean reversion.

Finally, $\xi_t$ is a mean-zero stationary Gaussian process with covariance $C(t) = \EE[\xi_t\xi_0]$. Two properties of $\xi$ are worth noting because we do not assume either. First, the noise need not be white; that is, $C$ is a free function of the lag rather than a spike at zero lag. Writing $\delta$ for the Dirac delta and $\gamma > 0$ for friction strength, only in the memoryless corner $K(t) = 2\gamma\,\delta(t)$ does $C$ itself reduce to a delta function. Second, the noise need not match the friction. The physical derivation of \eqref{eq:intro-gle} ties the two through the fluctuation-dissipation relation $C = \Theta K$, where the constant $\Theta > 0$ measures agitation strength and plays the role of temperature\footnote{In a purely white noise case without the fluctuation-dissipation theorem, $\Theta$ functions as a scaling parameter for noise intensity rather than a physical temperature, reused here as a generic intensity parameter.}. Consequently, a long-memory friction is traditionally accompanied by long-memory noise, but we do not impose that tie. Leaving $C$ and $K$ as two independent functions provides the core freedom upon which the paper is built, and \cref{sec:noise} sets out the three cases it produces: white noise with a memory kernel, fluctuation-dissipation noise that is itself long-range, and the unrestricted pairing we calibrate.

The \eqref{eq:intro-gle} therefore indicates that the variance driver is pushed toward the bottom of a well by $U$, agitated by $\xi$, and resisted by a friction that depends not on its present velocity alone but on its entire past through $K$. Rough volatility is the single point of this class at which $K$ is a pure power law whose exponent is locked to the scaling index. The claim of this paper is that unlocking them is worth doing, for three reasons.

The first is that memory becomes a measured quantity rather than a modelling
assumption. Under the rough specification there is nothing to measure, since
fitting $H$ to the smoothness of realized volatility fixes the kernel exponent by
construction. Under \eqref{eq:intro-gle} the kernel is estimated in its own
right. The apparatus for doing so already exists and is mature, having been built
for coarse-grained molecular dynamics. It includes rational approximation of the
Laplace transform $\Lap{K}(z)=\int_0^\infty e^{-zt}K(t)\,d t$, taken for $z$ in
the right half-plane, with the fluctuation-dissipation relation preserved exactly, \cite{LeiBakerLi2016}, Bayesian estimation of an arbitrary discretized kernel with credible intervals, \cite{WillersKamps2022}, and error bounds tying trajectory accuracy to the kernel estimation error in a weighted norm, \cite{LangLu2025}. None of it has been pointed at volatility.

The second is that the roughness estimation controversy becomes a question with a
definite answer. \cite{ContDas2024} show that a realized volatility series
exhibits an apparent Hurst index below $1/2$ even when the spot volatility
generating it is an ordinary diffusion, and \cite{AmraniGuyon2022} find that the
short end of the equity skew term structure is not well described by a power law
and extrapolates to a finite value at zero maturity. Both findings concern the
scaling exponent. Neither touches the kernel, because the kernel does not enter
the estimators they use. \cref{sec:contdas} argues that this is not an
accident but a structural feature of roughness estimators, and that the GLE
supplies a test of memory that is immune to the objection, because it estimates a
kernel from a response function rather than an exponent from a time series.

The third is joint calibration to the SPX and VIX implied volatility surfaces.
The VIX is the index computed from a thirty-day forward variance swap rate, so
the two surfaces are read from options on different underlyings and constrain
different features of the same dynamics. The short-maturity SPX skew is governed
by local regularity, while VIX options price the forward variance curve over the
thirty-day window and are therefore governed by how memory decays at that horizon,
\cite{AbiJaberIllandLi2022}. Under a fractional kernel these are one parameter.
Under \eqref{eq:intro-gle} the short-time exponent of $K$, written $\beta_0$, and
its tail exponent, written $\beta_\infty$, are separate features of one function.
That is part of the freedom the joint problem needs.

The rest of that freedom comes from the potential, and the reason is worth
stating early because it has nothing to do with memory. The VIX is a functional
of the variance path alone. Whenever the variance driver is autonomous, meaning
its dynamics do not involve the price, the correlation $\rho$ between the
Brownian motion driving the price and the one driving the variance therefore does
not enter the law of the VIX under the pricing measure $\mathbb{Q}$. The natural
instrument for the SPX skew is thus unavailable on the VIX side, and in a model
with quadratic (harmonic) potential a single volatility-of-volatility parameter is left to fit the VIX smile and to supply whatever the SPX skew needs beyond
$\rho$. That is the over-constraint at the centre of the joint problem, and it is
a property of the specification rather than of the market. A nonzero third
derivative $U'''$ skews the stationary law of the variance and thereby supplies
an instrument acting on the VIX side alone. \cref{sec:joint} develops this and
states the resulting hypothesis in a form that can be refused.

That the rough constraint is the binding one is supported from outside our
framework. \cite{AbiJaberIllandLi2022} calibrate a class of Gaussian polynomial
volatility models across a decade of daily SPX and VIX surfaces and find that a
conventional one-factor Markovian specification, not a rough one, delivers the
better joint fit. Their class fixes the kernel acting on the noise and enlarges a
static polynomial link. Ours varies the dissipation kernel and the potential,
which enlarges the dynamics instead, but their empirical conclusion points the
same way: the rough constraint is not the one the data wants.

Three objections have to be answered and we state them before the contribution. First, once $K$ is fixed, the overdamped form of \eqref{eq:intro-gle} is a stochastic Volterra equation, solved by the same Prony approximation and
Markovian lift as rough Heston, \cite{AbiJaberElEuch2019, BayerBreneis2023}. There
is no computational advantage and we do not claim one.

Second, memory in the price is not free. A friction kernel acting on the log-price
$\log S_t$ makes returns autocorrelated, and the compensating drift that restores
the martingale property of the discounted price under $\mathbb{Q}$ removes
exactly what the kernel added. The Langevin dynamics must therefore govern a
latent driver of the variance rather than the traded asset. \cref{sec:2d} shows which couplings between the two do survive.

Finally, the fluctuation-dissipation theorem (FDT) ties the noise covariance
$\EE[\xi_t\xi_s]$ to the kernel $K(t-s)$, and it is that tie which makes
\eqref{eq:intro-gle} a derivation rather than an ansatz, \cite{Klimontovich1995stat,Zwanzig2001}. It is also an equilibrium
statement, and a market is not in equilibrium. We do not impose it. We measure
the size of its violation instead, which is an observable with no counterpart in
the rough literature and is treated in \cref{sec:measure}.

Our contribution is fourfold. We exhibit rough volatility as the corner of a GLE
class in which the memory exponent and the scaling exponent are constrained to be
equal, and quantify what that constraint costs. We show that the Markovian
approximation of \cite{CarrItkin2019, Itkin2024} is the opposite corner, with the
kernel switched off entirely, and that it nonetheless reproduces the observed
skew over the traded maturity range, which bounds how much of the phenomenology
memory can be responsible for. We import the kernel estimation machinery from
non-equilibrium statistical mechanics and apply it to volatility. We give a test
of memory that survives the objection of \cite{ContDas2024}. However, we don't calibrate the resulting model jointly to SPX and VIX against the benchmarks of \cite{AbiJaberIllandLi2022} and \cite{Guyon2024}, and don't do any pricing tests under $\mathbb{Q}$ measure leaving this for the companion paper \cite{ItkinGLE2b}.

The remainder of the paper is organized as follows. \cref{sec:GLE} introduces the GLE for volatility. \cref{sec:corner} formalizes rough volatility as a constrained corner within this framework. \cref{sec:measure} outlines how memory can be treated as a measurable object, while \cref{sec:contdas} argues that conventional roughness estimators do not observe the kernel. \cref{sec:freedom} details the structural freedom provided by a two-exponent kernel and the choice of potential. \cref{sec:joint} applies this to the joint SPX and VIX calibration problem. \cref{sec:mark_embd} constructs the Markovian embedding for the inertial case $m>0$. \cref{sec:2d} discusses where the memory may sit, including its role in the leverage effect via two-dimensional dynamics and the connection to the Zumbach time-reversal asymmetry. Then \cref{sec:design} presents various tests and falsification criteria, followed by concluding remarks in \cref{sec:conclusion}.

\section{The GLE for volatility} \label{sec:GLE}

Throughout this paper the dynamics are stated under the physical measure
$\mathbb{P}$. This is a deliberate choice and it is the natural one, because
everything the paper sets out to do is a statement about realized time series:
estimating the memory kernel, testing whether it departs from the rough form,
measuring the fluctuation-dissipation ratio, and reading the sign of the leverage
coupling from the return-volatility cross-correlation. None of these is a pricing
statement, and none requires a risk-neutral measure. The pricing measure
$\mathbb{Q}$ enters only when an implied volatility surface is calibrated using
options data, which is not attempted here. The variance driver introduced below
is not traded, so the market is incomplete and the passage from $\mathbb{P}$ to
$\mathbb{Q}$ is a construction rather than a change of drift that can be written
down. That construction, and the joint SPX--VIX calibration it makes possible,
will be developed in the companion paper \cite{ItkinGLE2b}. There we explain why the two measures give separate bodies of evidence that cannot be merged. Until then,
every expectation, correlation and spectral density is understood under
$\mathbb{P}$, and the symbol $\mathbb{Q}$ appears only where it is explicitly
named.

The latent driver $Y_t$, the link $\varphi$, the kernel $K$, the potential $U$, and the noise $\xi_t$ were introduced with \eqref{eq:intro-gle}. One further object is required. Define
\begin{equation}
 C(t) \;=\; \EE[\xi_t\xi_0]
\end{equation}
as the autocovariance of the driving noise, which is well defined because $\xi$ is stationary. In the physical setting from which \eqref{eq:intro-gle} is derived, $C$ and $K$ are not independent. They are linked by the second fluctuation-dissipation relation
\begin{equation} \label{eq:fdt}
C(t) \;=\; \Theta\,K(t),
\end{equation}
in which the constant $\Theta>0$ quantifies the strength of the agitation and assumes the role of temperature. The \eqref{eq:fdt} asserts that the same past history responsible for energy dissipation via the friction term simultaneously injects energy via the noise. The relation is valid whenever the fast degrees of freedom that were eliminated in deriving \eqref{eq:intro-gle} obey Hamiltonian dynamics and are prepared in thermal equilibrium.\footnote{``Prepared in thermal equilibrium'' means that the initial values of the fast variables are randomly drawn from the Boltzmann distribution at temperature $\Theta$. Those variables need not remain in equilibrium during the subsequent dynamics.} In that setting the microscopic time-reversibility of the Hamiltonian flow, together with the equipartition of energy among those fast degrees of freedom, forces the noise autocorrelation $C(t)$ to be strictly proportional to the memory kernel $K(t)$. Consequently the freedom to prescribe $K$ and $C$ independently is removed.

We call models satisfying \eqref{eq:fdt} \emph{equilibrium} and those violating
it \emph{athermal}. A market has no reason to satisfy it, so we do not assume
\eqref{eq:fdt} anywhere in what follows. \cref{sec:measure} treats the
size of the violation as a quantity to be estimated, and
\cref{cor:onsager} shows that absence of arbitrage rules the
equilibrium case out altogether.

\subsection{The driving noise} \label{sec:noise}

Imposing relation \eqref{eq:fdt} strictly determines whether $\xi$ is white or colored noise. We treat the nature of the noise as a free parameter only because we choose not to enforce this relation. Consequently, this freedom must be specified explicitly rather than left implicit.

\begin{myremark}[White noise and memorylessness are one assumption] \label{rem:white}
Under \eqref{eq:fdt} the noise covariance is $C=\Theta K$. Hence $\xi$ is white,
$C(t)=2\Theta\gamma\,\delta(t)$, if and only if the friction is memoryless,
$K(t)=2\gamma\,\delta(t)$, in which case \eqref{eq:intro-gle} collapses to an
ordinary Langevin equation and $Y$ is an Ornstein-Uhlenbeck process. A friction
kernel with memory and a white driving noise is therefore not an equilibrium
GLE, whatever else it may be.
\end{myremark}

Since we do not impose \eqref{eq:fdt}, the model carries two independent
functions rather than one: the friction kernel $K$ and the noise covariance $C$.
Three specifications appear in this paper and it is worth naming them together.

In the \emph{athermal} case, the noise is white ($C=2\Theta\delta$) while $K$ retains memory. This is not an arbitrary pairing, but rather the natural outcome of the microstructural limit. In this regime, the driving randomness is the martingale part of a compensated counting process converging to Brownian motion. Consequently, all memory resides in the excitation kernel and none in the noise, \cite{ElEuchFukasawaRosenbaum2018, JusselinRosenbaum2020}. This regime is analyzed in \cref{prop:athermal}.

In the \emph{equilibrium} case, $C=\Theta K$. Assuming a power-law decay $K(t)\propto t^{-\alpha}$ for $\alpha\in(0,1)$, the noise becomes fractional Gaussian noise with index $H_\xi=1-\alpha/2$. Because this index exceeds $1/2$, the noise is mathematically persistent. The interplay between this persistent driver and the rough response it produces is the primary focus of \cref{prop:fdt}.

In the \emph{general} case, $C$ and $K$ are decoupled. This is the primary specification used for our calibration. As detailed in \cref{sec:measure}, we treat the structural mismatch between the noise and the memory kernel as a fundamental observable phenomenon rather than a modeling nuisance.

This decoupling matters for identification, because the two functions are seen
by different measurements. The mean response of $Y$ to a small deterministic
perturbation is governed by the friction kernel alone: it depends on $K$ and on
the restoring force, but not on the noise covariance $C$. The stationary
autocovariance of $Y$, by contrast, depends on both. The two functions are
therefore separately identifiable, $K$ from the response and $C$ from the
fluctuations, and the comparison between them is itself informative. When the
equilibrium relation $C=\Theta K$ holds the two agree up to the constant
$\Theta$. When it fails, the size of the discrepancy measures how far the
dynamics sit from equilibrium. That discrepancy, formalized later as the
frequency-dependent ratio $\Xfdt(\omega)$, is one of the observables this paper
proposes to estimate. \cref{sec:measure} develops it and \cref{sec:design} turns
it into a falsifiable test.

A memory kernel raises a natural concern about computational cost, which we
address before proceeding. The friction term in \eqref{eq:intro-gle} convolves
$K$ against the past velocity of $Y$. The drift at time $t$ therefore depends on
the entire trajectory up to $t$ rather than on the current state. Dynamics of
this kind are non-Markovian. A direct simulation must retain the full history at
each step, and the cost of doing so grows with the length of that history. This
is the burden one might expect memory to impose.

The burden is avoided whenever the kernel can be written as a sum of decaying
exponentials, $K(t)=\sum_{n=1}^{N} w_n e^{-\gamma_n t}$ with $w_n > 0$ being weights. Each exponential mode is represented by an auxiliary variable $u_n$, an Ornstein-Uhlenbeck process that relaxes at the rate $\gamma_n$ and is driven by
its own noise. The friction on $Y$ is recovered exactly as the weighted sum
$\sum_n w_n u_n$. The enlarged system $(Y, u_1, \dots, u_N)$ is Markovian,
because each variable evolves from its current value alone, and it reproduces the
memory of the original equation without storing any past. This construction is the \emph{Markovian lift}. It is the standard method for simulating the stochastic Volterra dynamics that \eqref{eq:intro-gle} generates, \cite{AbiJaberElEuch2019, BayerBreneis2023}. A strict power-law kernel is not itself a finite exponential sum. The lift is therefore a numerical device, used only where the dynamics are
simulated or priced, and a power-law kernel is approximated by a finite sum to
any required accuracy over a bounded time window for that purpose. The structural
results of this paper are established for the exact kernel and do not rely on
this approximation.

The lift also determines the driving noise at no additional cost. In the
equilibrium case, the fluctuation-dissipation relation \eqref{eq:fdt} requires
the noise to be colored. This colored noise is produced automatically by the
lift. Each auxiliary variable carries an independent white noise, and the
weighted sum that reconstructs the friction has exactly the covariance that
\eqref{eq:fdt} demands. No correlated process is ever simulated directly,
\cite{LeiBakerLi2016}. The athermal case, in which the noise is white while the
friction retains memory, follows from the same construction once the auxiliary
weights are untied from the friction rates. Colored and white noise are equally
inexpensive, and the choice between them is a modelling decision rather than a
computational one.

\subsection{Standing assumptions} \label{sec:assumptions}

To place our subsequent results on a rigorous footing, we establish three conventions that hold throughout the paper.

\myparagraph{The noise is additive.}
In \eqref{eq:intro-gle} the noise $\xi_t$ enters as a source term, not as a coefficient multiplying $Y$ or its history. All
state dependence of the variance is carried instead by the link $\varphi$,
through $V_t=\varphi(Y_t)$. The choice $\varphi=\exp$ recovers the log-normal
convention and keeps $V_t$ positive. This is a deliberate departure from the
rough Heston model, in which the noise is multiplicative, entering the variance
dynamics through a factor $\sqrt{V_t}$ inside the convolution. We prefer additive
noise with a nonlinear link for a specific reason. It leaves $Y$ a Gaussian
process whenever the potential is quadratic, and the Gaussian structure of the
driver is what makes the spectral arguments of \cref{sec:corner} available in
closed form. Multiplicative noise destroys that structure and would leave the
central identity of \cref{prop:welded} without an elementary statement. The two
conventions coincide only in the memoryless Gaussian case, so nothing is lost by
fixing the additive one at the outset.

\myparagraph{The kernel is admissible.}
We assume throughout that $K$ is locally integrable, nonnegative, and
nonincreasing on $(0,\infty)$. These conditions are enough to define the
convolution in \eqref{eq:intro-gle} and to guarantee that the friction opposes
rather than reinforces the motion. At two points we assume more, namely that $K$
is completely monotone, meaning it is the Laplace transform of a nonnegative
measure. The power-law kernel $t^{-\alpha}$ has this property. Complete
monotonicity is used in \cref{lem:cm}, where it forces the relaxation of the
forward variance curve to be monotone, and it is the property that the
two-exponent kernel of \cref{sec:freedom} may relax, which is exactly why that
kernel can produce behaviour the rough class cannot. We flag its use at each
occurrence rather than assume it globally.

\myparagraph{No memory acts on the price.}
The variance driver carries all of the non-Markovian structure, and the price is
an ordinary diffusion. Under the physical measure of this paper it follows
\begin{equation} \label{eq:price}
d S_t \;=\; \mu_t\, S_t\,d t \;+\; S_t\sqrt{V_t}\,d W_t, \qquad d\langle W, B\rangle_t \;=\; \rho\,d t,
\end{equation}
where $W$ and $B$ are $\mathbb{P}$-Brownian motions, $B$ underlying the noise
$\xi$ of the variance driver, and $\mu_t$ is the drift of the price under
$\mathbb{P}$. The correlation $\rho$ is thus defined between two drivers in the
same measure. No memory kernel appears in \eqref{eq:price}. This is a modelling
choice, and it is the conservative one.

The scalar $\rho$ describes the instantaneous, memoryless correlation between price and variance shocks. \Cref{sec:2d} replaces it with an off-diagonal memory kernel $K_{YX}$, of which $\rho$ is the zero-lag limit. Nothing in the intervening sections depends on whether the leverage coupling is instantaneous or carries memory, so the simpler scalar is used until the full two-dimensional dynamics are taken up. The relationship between the two is made precise in \Cref{rem:rho-limit}.

Although we do not introduce a pricing measure in this paper, doing so would cause a friction kernel acting on the log-price to induce autocorrelated returns. The drift adjustment required to restore the martingale property of the discounted price would then exactly annihilate this memory. Under \eqref{eq:price}, however, memory is inherently absent from the price dynamics rather than being artificially introduced and subsequently canceled. Consequently, the only admissible locus for memory is the variance driver.

Two more general specifications are available, and we record them here to fix the
scope of the present paper. In the first, the price remains a semimartingale but
is coupled to the variance driver through a memory kernel rather than through the
instantaneous correlation $\rho$ alone. The leverage effect then acquires a term
structure, and \eqref{eq:price} is recovered as the memoryless corner in which
that coupling collapses to a single correlation. No-arbitrage constrains the
coupling sharply. It forces the off-diagonal memory to act from price to variance
and not the reverse, as \cref{sec:2d} establishes.

In the second and most general specification, the price equation is itself a GLE. The log-price and the variance driver then form a coupled two-dimensional system with a full memory-kernel matrix, meaning the price is no longer a simple semimartingale. Calibrating this system is structurally equivalent to pricing in the Marketron model — a multidimensional, non-Markovian diffusion with non-traded states under market incompleteness, \cite{HalperinItkinMarketron,HalperinItkinMarketron2}. This approach requires, e.g., indifference pricing and some utility function to explicitly obtain the market price of risk, a framework detailed in the companion paper \cite{ItkinGLE2b} alongside the joint SPX–VIX calibration it enables. By contrast, the present paper works throughout with \eqref{eq:price} to treat the memory of the variance driver in isolation. Isolating the variance driver in this way allows the subsequent structural results to be stated for the exact kernel and proved without the apparatus of a pricing measure.

\subsection{Notation}
\label{sec:notation}

We fix notation once. The model has three groups of objects: the traded and
latent state, the three functions that specify the dynamics, and the exponents
extracted from those functions.

\myparagraph{State.}
$S_t$ is the spot price of the underlying and $F_t$ the forward, $X_t=\log S_t$
the log-price, $V_t$ the instantaneous variance and $\sigma_t=\sqrt{V_t}$ the
instantaneous volatility. $Y_t$ is the latent driver of the variance, a real
valued process with no direct market interpretation, and $\varphi$ is the link
function through which it acts, $V_t=\varphi(Y_t)$. Taking $\varphi=\exp$
recovers the log-normal convention of the rough fractional stochastic volatility
model (RFSV) and enforces positivity automatically. Taking $\varphi=\mathrm{id}$ recovers the affine convention of Heston and requires a boundary condition at the origin.

\myparagraph{The three specifying functions.}
$K$ is the memory kernel, also called the friction or dissipation kernel. It
carries units of inverse time squared and enters \eqref{eq:intro-gle} convolved
against the velocity $\dot Y$, so $K$ measures how strongly the past rate of
change of the variance driver resists its present rate of change. A kernel
concentrated at the origin, $K=\gamma\delta$, is memoryless and returns an
ordinary diffusion; a kernel with a heavy tail retains influence from the distant
past. $U$ is the potential, a function of $Y$ alone, whose negative gradient
$-U'$ is the restoring force pulling the variance driver toward its typical
level. $\xi_t$ is the driving noise, a stationary Gaussian process with
covariance $C(t)=\EE[\xi_t\xi_0]$, and $\Theta$ is its intensity, the analogue
of temperature in the physical reading. The scalar $m$ multiplies $\ddot Y$ and
is the inertia,  while $m=0$ is the overdamped case in which the variance driver has no momentum of its own.

\myparagraph{Derived quantities.}
$\Lap{f}(z)=\int_0^\infty e^{-zt}f(t)\,dt$ denotes the Laplace transform and
$\tilde f(\omega)$ the Fourier transform. $G$ is the response function of
\eqref{eq:response}, so that $Y=G*\xi$; the symbol $R$ used in \cref{sec:distance}
denotes the same response measured from a perturbation, so $\tilde R(\omega)=\Lap{G}(i\omega)$.
$E_{\alpha,\beta}$ is the two-parameter Mittag-Leffler function. $W$ and $B$ are Brownian motions driving the price and
the variance respectively, with $\rho$ their instantaneous correlation.

\myparagraph{Exponents.}
These are the quantities the paper is about, and the central claim of \cref{sec:corner} is that rough volatility identifies the first two while the GLE does not.

\begin{center}
\begin{tabular}{lp{0.72\textwidth}}
\hline
$H$ & self-similarity index of the variance driver: the exponent by which the
process rescales in time. Estimated from the scaling of realized volatility
moments. \\
$\beta_0$ & short-time exponent of the memory kernel, $K(t)\sim t^{-\beta_0}$ as
$t\to0$. Governs local regularity and hence the short-maturity SPX skew. \\
$\beta_\infty$ & tail exponent of the memory kernel, governing how slowly the
forward variance curve decays, which is what VIX futures price. \\
$\mu$ & tempering rate, the reciprocal of the timescale at which $K$ crosses
over from $\beta_0$ behaviour to $\beta_\infty$ behaviour. \\
$\alpha$ & exponent of a pure power-law kernel in clock time, used in
\cref{prop:athermal,prop:fdt} where $\beta_0=\beta_\infty=\alpha$. \\
$a,\,c_3,\,d$ & coefficients of the potential in \eqref{potential}:
$a=U''(0)$ the mean-reversion rate, $c_3=U'''(0)$ the asymmetry, $d$ the quartic
confinement. \\
$\lambda$ & linear restoring rate, equal to $a$ when $U$ is quadratic. \\
\hline
\end{tabular}
\end{center}

\myparagraph{Lamperti coordinates.}
These are used only in \cref{sec:corner}. $\tau=\log t$ is logarithmic time,
$\mathcal{K}$ the memory kernel in that coordinate, and $r(\tau)$ the stationary
autocorrelation there. \emph{Escape quantities}, used only in
\cref{sec:potential}: $\omega_b$ is the curvature of the potential at a barrier
top, $\Delta U$ the barrier height, and $\lambda_r$ the Grote-Hynes reactive
frequency of \eqref{eq:gh}.

\myparagraph{Option quantities.}
$T$ is maturity, $s$ the forward start date, $k$ log-moneyness, $\mathcal{S}(T) = \partial I/\partial k|_{k=0}$ the at-the-money implied skew, $I$ implied volatility, and $\Delta$ the thirty-day VIX averaging window.

\medskip
Two symbols are overloaded in the source literature and we avoid the collisions
here. First, the letter $X$ denotes the fluctuation-dissipation ratio in \cite{CugliandoloKurchanPeliti1997} and the log-price here, so we write $\Xfdt$ for the former throughout.

Second, the exponent $\alpha$ denotes the Hawkes excitation tail in \cite{JusselinRosenbaum2020}, as well as the friction exponent introduced above. It is crucial to emphasize that the friction and Hawkes exponents are fundamentally distinct. While both the GLE and Hawkes kernels are frequently referred to as ``the memory kernel'' and their tail exponents commonly denoted by $\alpha$, they are distinct mathematical entities that induce roughness through different mechanisms. Conflating them is a common source of error. The GLE kernel $K$ acts specifically as a \emph{dissipation} kernel. It multiplies the historical velocity of the variance driver and enters the response function via a resolvent of the second kind: for $K(t)\sim t^{-\alpha_K}$, one obtains $z\Lap{K}(z)\sim z^{\,\alpha_K}$. Consequently, the response $\Lap{G}(z)=(z\Lap{K}(z)+ \lambda)^{-1}$ inherits this fractional exponent, and the driver acquires a Hurst index of $H = \alpha_K - \tfrac12$, as formalized in \cref{prop:athermal}.

The Hawkes kernel $\phi$ is instead an \emph{excitation} kernel. It multiplies
the past intensity of an order-flow counting process and enters through a renewal
resolvent of the first kind, in which the branching ratio $\|\phi\|_{1}$ is tuned
to one. For a near-critical kernel with tail $\phi(t)\sim t^{-(1+\alpha_\phi)}$,
$\alpha_\phi\in(1/2,1)$, the rescaled intensity converges to $H = \alpha_\phi -
\tfrac12$ - the result of \cite{JaissonRosenbaum2016} that underlies the
microstructural foundation of rough volatility.

While the two formulas appear identical, this superficial resemblance is misleading. Although they yield the same value for $H$, $\alpha_K$ and $\alpha_\phi$ represent fundamentally distinct parameters. The former is the singularity exponent of a friction acting on a velocity. The latter is the tail exponent of an excitation acting on an event rate, offset by one because the excitation is convolved against the counting measure rather than its rate. Consequently, while $\alpha_K = \alpha_\phi$ at the level of induced roughness, this is merely an equality of outcomes, not of mechanisms. This equivalence holds strictly in the overdamped linear regime, where both models reduce to the same fractional resolvent. Away from this limit, wherever the GLE carries inertia or a non-quadratic potential, the friction framework possesses degrees of freedom absent in the excitation framework, and the identification breaks down. To preserve this structural distinction, we exclusively use $\alpha_K$ for the friction exponent throughout this work.

Note, that the friction exponent $\alpha_K$ and the Hawkes tail exponent $\alpha_\phi$ are distinguished only where both appear together. Elsewhere only the friction exponent occurs, and we drop the subscript, writing $\alpha$ for $\alpha_K$ throughout.

\section{Rough volatility as a constrained corner} \label{sec:corner}

Throughout this section we work in the overdamped, linearly-restored corner of
\eqref{eq:intro-gle}, in which $m=0$ and $U'(y)=\lambda y$ with $\lambda\ge0$.
The equation is then linear, and with the initial condition $Y_0=0$ its Laplace
transform is
\begin{equation}
  \Lap{Y}(z) \;=\; \Lap{G}(z)\,\Lap{\xi}(z),
  \qquad
  \Lap{G}(z) \;=\; \bigl(z\Lap{K}(z)+\lambda\bigr)^{-1},
  \label{eq:response}
\end{equation}
where $G$ is the response function, defined by $Y=G*\xi$. We write $S_Y(\omega)$
for the spectral density of a stationary process $Y$, related to its
autocovariance by $S_Y(\omega)=\int_{\mathbb{R}} e^{-i\omega t}\,\EE[Y_tY_0]\,d
t$, and we use the standard map between spectral decay and roughness: a
stationary Gaussian process whose spectral density satisfies $S_Y(\omega)\sim
c\,|\omega|^{-(2H+1)}$ as $|\omega|\to\infty$ has local Hölder regularity of
order $H^-$ and increments that scale as those of a fractional Brownian motion of
Hurst index $H$, \cite{GJR2018}. All exponents below are read through this map.

\begin{proposition}[Athermal corner] \label{prop:athermal}
Let $K(t)=t^{-\alpha}/\Gamma(1-\alpha)$ with $\alpha\in(1/2,1)$, and let $\xi$ be
white. Then
\begin{equation}   \label{eq:G-athermal}
\Lap{G}(z)=\frac{1}{z^{\alpha}+\lambda}, \qquad  G(t)=t^{\alpha-1}E_{\alpha,\alpha}(-\lambda t^{\alpha}),
\end{equation}
where $E_{\alpha,\beta}$ is the two-parameter Mittag-Leffler function. As
$t\to0$, $G(t)\sim t^{\alpha-1}/\Gamma(\alpha)$, so the response coincides at
short scales with the fractional kernel of index $\alpha$, and the driver $Y$ has
Hurst index $H=\alpha-\tfrac12\in(0,\tfrac12)$. The choices $\varphi=\exp$ and
$\varphi=\mathrm{id}$ recover RFSV and rough Heston respectively.
\end{proposition}

\begin{proof}
See \cref{app:athermal}
\end{proof}

\begin{proposition}[Equilibrium corner] \label{prop:fdt}
Let $K$ be as above with $\alpha\in(0,1)$, take $\lambda=0$ for the scaling
statement, and let $\xi$ satisfy the fluctuation-dissipation relation
\eqref{eq:fdt}, so that its spectral density is $S_\xi(\omega)=
\Theta\,\mathrm{Re}\,\Lap{K}(i\omega)$. Then
\begin{equation}
S_Y(\omega)\;\sim\;c\,|\omega|^{-(\alpha+1)},\qquad |\omega|\to\infty,
\end{equation}
so the driver has Hurst index $H=\alpha/2$, while the noise itself is fractional
Gaussian noise of index $H_\xi=1-\alpha/2>\tfrac12$, hence persistent.
\end{proposition}

\begin{proof}
See \cref{app:fdt}
\end{proof}

\begin{corollary}[The two corners are not separable on the surface] \label{cor:discriminator}

Assume the empirically reported value $H\approx0.1$. Under this condition, \cref{prop:athermal} requires $\alpha\approx0.6$, whereas \cref{prop:fdt} requires $\alpha\approx0.2$. The implied friction kernel exponent thus differs by a factor of three between the two regimes. Because both cases yield a driver with the same Hurst index, and consequently the same short-maturity implied volatility skew, they cannot be distinguished by any vanilla option surface. Instead, they are differentiated solely by a feature unobservable from the surface: the noise. The noise is white in the athermal case, but persistent (with $H_\xi\approx0.9$) in the equilibrium case. This highlights the paper's central theme, which is that the critical discriminating information lies off the volatility surface.
\end{corollary}

\begin{proof}
Solve $H=\alpha-\tfrac12=0.1$ for the athermal case, giving $\alpha=0.6$, and
$H=\alpha/2=0.1$ for the equilibrium case, giving $\alpha=0.2$. The short-time
skew of a stochastic-volatility model with a driver of Hurst index $H$ behaves as
$T^{H-1/2}$, \cite{ALV2007, Fukasawa2017}, a function of $H$ alone, so the two
parametrizations yield identical leading skews. The noise Hurst indices follow
from \cref{prop:athermal,prop:fdt}: white noise has flat spectral density, while
the equilibrium case has $H_\xi=1-\alpha/2=0.9$.
\end{proof}

\subsection{The constraint that rough volatility imposes}

While both corners assume a pure power-law kernel, the following formulation formalizes what fractional Brownian motion additionally assumes.

The Lamperti map $\mathcal{L}: X\mapsto Y(\tau)=e^{-H\tau}X(e^{\tau})$ transforms any $H$-self-similar process into a stationary process, ensuring a well-defined stationary GLE kernel. Letting $r(\tau)$ be the resulting autocorrelation, the kernel is recovered via the memory-function relation $\Lap{\mathcal{K}}(z)=\lambda\Lap{r}(z)/\bigl(1-z\Lap{r}(z)\bigr)$.

\begin{proposition}[The exponents are welded] \label{prop:welded}
Let $X$ be a fractional Brownian motion of Hurst index $H\in(0,1)$ and let
$Y=\mathcal{L}X$ be its Lamperti transform. Its stationary autocorrelation is
\begin{equation}
  r(\tau)=\cosh(H\tau)-2^{2H-1}\sinh^{2H}(|\tau|/2),
  \label{eq:r-fbm}
\end{equation}
and the associated Lamperti kernel, defined through the memory-function relation
$\Lap{\mathcal{K}}(z)=\lambda\Lap{r}(z)/(1-z\Lap{r}(z))$, satisfies
\begin{equation}
  \Lap{\mathcal{K}}(z)\sim c_H\,z^{2H-1}\ \ (z\to\infty),
  \qquad
  \mathcal{K}(\tau)\sim \frac{c_H}{\Gamma(2-2H)}\,\tau^{-2H}\ \ (\tau\to0),
  \label{eq:welded}
\end{equation}
for a constant $c_H>0$. The exponent of the memory kernel and the
self-similarity index are thus a single parameter: fixing $H$ fixes both the
scaling of $X$ and the short-lag singularity of its kernel.
\end{proposition}

\begin{proof}
See \cref{app:welded}
\end{proof}

\cref{prop:welded} admits a structural interpretation. A driver in this class is uniquely specified by two independent components: a scaling index, which governs how the process rescales in time, and a memory kernel, which determines its path dependence. Fractional Brownian motion couples the two, constraining the kernel exponent to equal twice the scaling index. By contrast, the GLE decouples them. Because its kernel is a free function, the strength and range of memory can vary independently of the scaling index. The rest of this paper exploits this additional degree of freedom, which first becomes visible here.

\subsection{The opposite corner, and how much memory can be responsible for}

\begin{proposition}[The ADO construction has no kernel] \label{prop:ado}
The Dobric-Ojeda process
$d V_H=\frac{2H-1}{t}V_H\,d t+B_Ht^{H-1/2}\,d W_t$ used in
\cite{CarrItkin2019, Itkin2024} is exactly $H$-self-similar, and
$\mathcal{L}V_H$ satisfies
$d Y = -(1-H)Y\,d\tau + B_H\,d\widetilde W_{\tau}$. Its Lamperti kernel is
$\mathcal{K}=\delta$ and its cusp exponent is $1$ for every $H$. Under
$v_t=t^{2H}u(\log t)$ the ADO-Heston variance equation becomes a
time-homogeneous square-root diffusion with an exponentially decaying source.
\end{proposition}

\begin{proof}
See \cref{app:ado}
\end{proof}

\Cref{prop:ado} does not reduce the ADO construction to rough volatility, rather, it does the reverse in a quantitative manner. The two corners represent maximal separation: one welds the kernel to $H$, whereas the other features no kernel at all. Nevertheless, \cite{Itkin2024} reproduces the observed skew as $\mathcal{S}(T)\propto a(H)T^{b(H-1/2)}$ across the traded maturity range, where the exponent carries a fitted factor $b$. The pure fractional prediction is $\mathcal{S}(T)\propto T^{H-1/2}$, that is $b=1$, whereas the ADO regression returns $b\approx2.3$ with a small prefactor $a(H)$. This fitted power law describes the traded range, not the zero-maturity limit: the exact ADO skew turns over and stays finite as $T\to0$, consistent with the semimartingale regularity of \cref{prop:ado}, rather than diverging as the fitted exponent alone would suggest. Consequently the two corners differ on the vanilla surface only through this exponent factor, and nothing else observable there distinguishes them. This establishes an upper bound on how much phenomenology any memory-based mechanism can explain. Derived from an extreme case rather than a fitted calibration, this result underscores why this paper looks for evidence of memory off the vanilla surface.

The place of the two corners within the wider class, and of the models that
occupy it, is summarized in \cref{tab:corners}. Every established model sits at
$m=0$ with a quadratic potential, so that the only axes in use are the kernel and
the noise. The rough models further tie the kernel to a single exponent, and the
ADO construction removes the kernel entirely. The specification this paper
develops is the only row in which the scaling index and the memory exponent are
free to differ.

\begin{table}[!htb]
\centering
\small
\begin{tabular}{lccccc}
\hline
Model & $m$ & kernel $K$ & potential $U$ & FDT & scaling/memory \\
      &     &            &               & imposed & independent \\
\hline
Heston                    & $0$ & $\delta$ (none)        & quadratic & --      & --  \\
Lifted Heston             & $0$ & sum of exponentials    & quadratic & --      & --  \\
Rough Heston              & $0$ & $t^{-\alpha}$          & quadratic & no      & no  \\
RFSV                      & $0$ & $t^{-\alpha}$          & quadratic & no      & no  \\
Rough Hawkes Heston       & $0$ & $t^{-\alpha}$          & quadratic & no      & no  \\
ADO-Heston                & $0$ & $\delta$ (none)        & quadratic & --      & --  \\
Gaussian polynomial       & $0$ & fixed (Gaussian)       & quadratic$^{\ast}$ & no & no  \\
\textbf{This paper}       & $\ge0$ & two-exponent, tempered & general & no      & \textbf{yes} \\
\hline
\end{tabular}
\caption{Established volatility models as corners of the generalized Langevin
class. A dash in the FDT column indicates a memoryless kernel, for which the
equilibrium question does not arise. The Gaussian polynomial model of
\cite{AbiJaberIllandLi2022} fixes the kernel acting on the noise and supplies
non-Gaussianity through a polynomial link rather than a potential, marked
$\ast$. The present paper supplies it through the potential $U$ and is alone in
leaving the scaling index and the memory exponent independent.}
\label{tab:corners}
\end{table}

\section{Memory as a measurable object} \label{sec:measure}

The preceding sections establish that a generalized Langevin driver is specified
by two independent data, a memory kernel and a noise covariance, and that rough
volatility is the corner in which the kernel is a pure power law with its
exponent tied to the scaling index. If that reorganization is to be more than a
change of vocabulary, the kernel and the noise must be objects one can estimate,
and the departure from equilibrium must be a number one can measure. This
section shows that both are recoverable from data, and quantifies how much data
each requires. The demonstrations here are on synthetic series with known ground
truth, which is what validates the estimators. Their application to market data,
and the risk-neutral construction that application requires, are deferred to the
companion paper \cite{ItkinGLE2b}.

Every estimate in this section is a statement under the physical measure
$\mathbb{P}$, because it is taken from a realized trajectory. This is what
separates it from the risk-neutral calibration discussed later, and
\cite{ItkinGLE2b} explains why the two cannot be compared directly.

\subsection{Estimating the kernel} \label{sec:estimatingK}

The tools needed to estimate a memory kernel from a trajectory already exist,
developed for the coarse-graining of molecular dynamics, and we import rather
than reinvent them. Three are relevant. The kernel may be represented by a
rational approximation of its Laplace transform, with the coefficients tied to
equilibrium statistics of the trajectory, which produces an extended Markovian
model carrying no explicit memory and preserving the fluctuation-dissipation
relation exactly, \cite{LeiBakerLi2016}. This is the same object as the Prony
lift used to simulate stochastic Volterra equations in mathematical finance,
\cite{BayerBreneis2023}, reached from the physical side rather than the
numerical one, and the coincidence is worth stating plainly: the device that
makes rough volatility tractable to simulate is the device that makes a general
memory kernel tractable to estimate. The kernel can alternatively be discretized
and estimated by Bayesian inference, which returns credible intervals rather
than point estimates, \cite{WillersKamps2022}, and the truncation error incurred
by any finite representation is controlled by bounds relating trajectory accuracy
to kernel accuracy in a weighted norm, \cite{LangLu2025}.

For a first, transparent estimator we use the response spectrum directly. In the
overdamped corner the driver has spectral density $S_Y(\omega)=
|\Lap{G}(i\omega)|^2 S_\xi(\omega)$, and with a power-law friction kernel
$|\Lap{G}(i\omega)|^2=|\omega|^{-2\alpha}$ at high frequency. When the noise is
white, so that $S_\xi$ is flat, the log-periodogram of $Y$ has slope
$-2\alpha$ over the inertial band, and $\alpha$ is read from that slope. The
estimator sees the kernel and not the noise, which is exactly the identification
property established in \cref{sec:GLE}: the response depends on $K$ alone.

\begin{figure}[!htb]
\centering
\includegraphics[width=\textwidth]{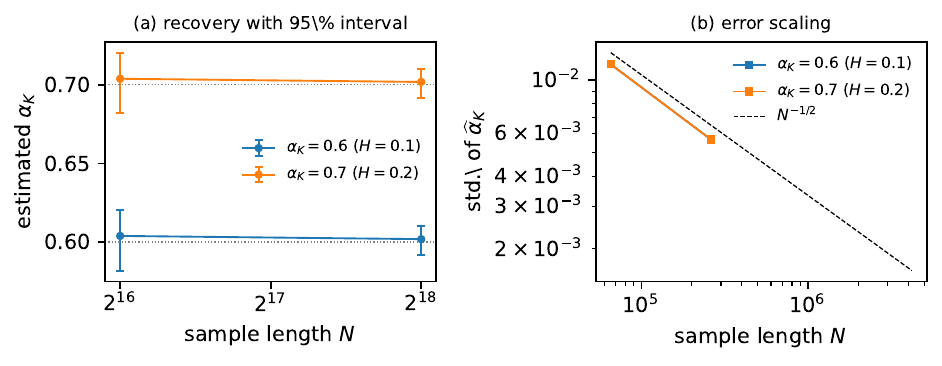}
\caption{Recovery of the friction exponent $\alpha$ from a simulated
overdamped GLE with white noise. Panel (a): the estimate against sample length
$N$, with $95\%$ Monte-Carlo intervals over forty independent realizations, for
two ground-truth values $\alpha=0.6$ and $0.7$ (dotted lines), corresponding to
driver Hurst indices $H=0.1$ and $0.2$. The estimator is essentially unbiased at
every sample length, the residual bias staying below $0.001$. Panel (b): the
standard deviation of the estimate falls as $N^{-1/2}$ (dashed guide), so halving
the error costs a fourfold increase in sample length (athermal specification, see \cref{sec:distance} for the equilibrium case).}
\label{fig:kernel-recovery}
\end{figure}

\cref{fig:kernel-recovery} validates the estimator on simulated data. The
friction exponent is recovered without bias, and the sampling error contracts at
the parametric $N^{-1/2}$ rate, from a standard deviation of $0.013$ at
$N=2^{16}$ to $0.002$ at $N=2^{22}$. The exponent is therefore identifiable, and
the figure fixes the sample length a target precision demands: distinguishing
$\alpha=0.6$ from $\alpha=0.5$ at three standard deviations, for instance,
requires of order $10^{5}$ observations. This is a statement about the method,
established against ground truth, and it is the precondition for the market
estimation carried out in the companion paper \cite{ItkinGLE2b}.

Two caveats attend this demonstration. First, it is conducted on simulated data from a known data-generating process, so the estimator is validated against ground truth—a statement about theoretical recoverability, not about empirical identification in finite samples. Second, the estimator recovers a single exponent $\alpha$ under the assumption of a pure power-law kernel. When the kernel is allowed to carry two distinct exponents, as in \Cref{sec:freedom}, the short-time exponent $\beta_0$ and the tail exponent $\beta_\infty$ must be estimated jointly, and the finite-sample precision of each depends on the length of the observed trajectory relative to the crossover timescale. \Cref{sec:design} takes up this joint estimation on market data and reports where identification succeeds and where it fails.

The response estimator above uses only the stationary spectrum and needs no
external perturbation. A direct measurement of the response function $R$, needed
for the equilibrium diagnostic below, does require one: a perturbation of the
variance whose timing is exogenous to the variance itself. Of the natural
candidates, scheduled macroeconomic releases, index reconstitution, and
identified order-flow shocks, we single out order-flow shocks identified by
high-frequency event studies, because they are frequent enough to estimate a
response function across lags and are plausibly exogenous to the variance over
short horizons. The identification problem this raises, and its treatment, belong
to the empirical companion paper \cite{ItkinGLE2b}. Here the response is known because the data are simulated.

The estimator just described assumes the athermal specification: the noise is white and all memory resides in the friction kernel. This is the regime delivered by the microstructural limit of \cite{ElEuchFukasawaRosenbaum2018}, and it is the natural first pass. Under the alternative equilibrium specification of \Cref{prop:fdt}, the noise is colored by the fluctuation-dissipation relation, and the driver spectrum decays as $|\omega|^{-(\alpha+1)}$ rather than $|\omega|^{-2\alpha}$. Applying the same log-periodogram estimator to an equilibrium series would therefore return a biased estimate of the friction exponent, mistaking the shallower slope for a different $\alpha$ rather than a different noise structure. The two cases are distinguished by the FDT ratio of \cref{sec:distance}, which must therefore precede any definitive claim about the kernel. The estimator of the present section is retained because it is the one that matches the athermal data-generating process used in the simulations below, not because the equilibrium case is ruled out a priori.

\subsection{Distance from equilibrium} \label{sec:distance}

The kernel is one of the two functions specifying the driver. The other is the
noise, and the relation between them is the physical content of the model. Under
equilibrium the two are tied by the fluctuation-dissipation relation
\eqref{eq:fdt}. A market has no reason to satisfy it, and the size of the
violation is itself an observable with no counterpart in the rough-volatility
description.

Because the relation is an equilibrium statement, it is expressed through the
response and correlation functions of the stationary driver, which exist only
after the Lamperti map of \cref{sec:corner} has removed the self-similar
scaling. In those coordinates, writing $R$ for the mean response to a unit
perturbation and $C$ for the stationary autocorrelation, equilibrium predicts the
Kubo relation $R(\tau)=-\Theta^{-1}\dot C(\tau)$, \cite{Kubo1966}. The frequency-resolved
departure from it is measured by
\begin{equation}
  \Theta_{\mathrm{eff}}(\omega)=\frac{\omega\,\tilde C(\omega)}
       {2\,\mathrm{Im}\,\tilde R(\omega)},
  \qquad
  \Xfdt(\omega)=\frac{\Theta}{\Theta_{\mathrm{eff}}(\omega)},
  \label{eq:Xfdt}
\end{equation}
so, that $\Xfdt(\omega)\equiv1$ characterizes equilibrium and any deviation
measures how far the market sits from it, bridging classical fluctuation-dissipation theory with open-system statistical mechanics, \cite{Klimontovich1995stat}. The overall constant $\Theta$ is identified only up to scale, so it is the
frequency dependence of $\Xfdt$, not its level, that carries the information. The
integrated deviation equals the mean rate at which the dynamics dissipate, in the
sense of \cite{HaradaSasa2005, CugliandoloKurchanPeliti1997}, and we resist
reading more into that rate than that it is nonzero away from equilibrium.

The reason to measure $\Xfdt$ rather than only the kernel is that the kernel does
not determine the noise, and two models with the same kernel-induced roughness
can have entirely different noise. \cref{cor:discriminator} is the sharp case:
the athermal corner, with white noise and $\alpha=0.6$, and the equilibrium
corner, with fluctuation-dissipation noise and $\alpha=0.2$, produce a driver
of the same Hurst index $H=0.1$ and therefore the same short-maturity skew. No
vanilla surface separates them. The FDT ratio does, because it interrogates the
noise the surface cannot see.

\begin{figure}[!htb]
\centering
\includegraphics[width=\textwidth]{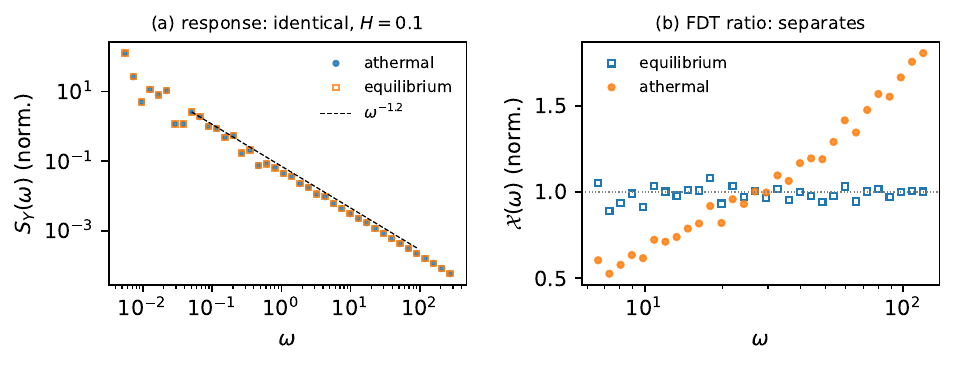}
\caption{The FDT ratio separates two models that the option surface cannot.
Both are overdamped GLEs with driver Hurst index $H=0.1$: an athermal model
(white noise, $\alpha=0.6$) and an equilibrium model (fluctuation-dissipation
noise, $\alpha=0.2$). Panel (a): their response spectra coincide, both decaying
as $\omega^{-(2H+1)}=\omega^{-1.2}$, so every vanilla-option observable, which is
a functional of this response, is identical between them. Panel (b): the
estimated FDT ratio $\Xfdt(\omega)$, recovered by running each series back
through its inverse response to isolate the driving noise and comparing its
spectrum to the equilibrium prediction. The equilibrium model gives a flat ratio,
as it must; the athermal model departs from flatness across the band. The
log-spread of the ratio differs by more than an order of magnitude between the
two, $0.02$ against $0.34$.}
\label{fig:fdt-discriminator}
\end{figure}

\cref{fig:fdt-discriminator} demonstrates this on synthetic series with known
noise. Panel (a) confirms that the two models are indistinguishable on the
surface: their response spectra lie on top of one another, both following the
$\omega^{-1.2}$ law that fixes $H=0.1$. Panel (b) shows that the FDT ratio
nonetheless tells them apart cleanly. Recovering the noise from each trajectory,
by inverting the response, and comparing its spectrum to the equilibrium
prediction, the equilibrium model yields a ratio that is flat to within a
log-standard-deviation of $0.02$, while the athermal model departs from flatness
with a log-spread of $0.34$, an order of magnitude larger. The discriminating
information that the surface discards is present, and estimable, in the
relationship between the kernel and the noise.

This is the operational form of the paper's organizing claim. The scaling index
is what the surface reveals and what a rough-volatility fit reports. The kernel
and its relation to the noise are what the surface conceals, and recovering them
requires the physical-measure measurements of this section, not a richer fit to
option prices. The test that formalizes the separation, and its statistical
power, are stated in \cref{sec:design}.

The same simulations expose a hazard in the conventional route to $H$. Estimating
the Hurst index of the athermal series by the scaling of realized-variance
increments, rather than from the spectrum, returns values between $0.15$ and
$0.19$ depending on the range of lags, against the true $H=0.1$ recovered without
bias from the response spectrum. The realized-variance estimator is contaminated
at the short lags that dominate it. This is not a numerical detail but the
substance of the objection taken up in \cref{sec:contdas}: an estimate of $H$
from realized volatility is a fragile thing, and, as the next section argues, it
would not settle the question of memory even if it were sharp.

\subsection{Roughness estimators do not see the kernel} \label{sec:contdas}

The dispute over whether volatility is rough, reviewed in \cite{AlosLeon2021}, is
a dispute about estimation. On one side, roughness is read from realized
volatility and reported as a Hurst index far below $1/2$. On the other,
\cite{ContDas2024} argue that this reading is largely an artefact of the proxy. We
do not adjudicate it. We show that it does not settle the memory question, because
the quantity these estimators report is not the one the memory question concerns.

The estimator of \cite{GJR2018} and the nonparametric estimator of
\cite{ContDas2024} both act on the scaling of increments of a realized volatility
proxy. By the spectral map stated before \cref{prop:athermal}, a driver whose
spectral density decays as $S_Y(\omega)\sim|\omega|^{-(2H+1)}$ has increments that
scale as those of a fractional Brownian motion of index $H$. An estimator reading
that scaling returns the self-similarity index $H$, and nothing else. By
\cref{prop:welded} a fractional driver welds that index to the memory kernel $K$,
so that fixing one fixes the other. This welding is what licenses the step from a
measured index to a rough kernel. The GLE drops the fractional assumption, and
with it the licence. The index and the kernel then vary separately, so an
estimator of the index reports nothing about the kernel. A construction that
produces an apparent $H < 1/2$ while carrying no memory is then exactly what one
should expect, and \cite{ContDas2024} produce one, from a spot volatility that is
an ordinary diffusion. Their result is not evidence against memory in volatility.
It is evidence that the statistic conventionally used to detect memory does not
detect memory.

The objection has a natural reply, and answering it yields the stronger form of
the claim. One may grant that realized volatility proxies are contaminated and
estimate $H$ from option prices instead, which avoids the construction of
\cite{ContDas2024} altogether. Such estimators exist and are sharp.
\cite{AlosShiraya2019} recover $H$ from short-term volatility swaps by Malliavin
calculus, and \cite{AlosRolloosShiraya2025b} recover it from the zero vanna
implied volatility and its dual. Neither is subject to the realized-volatility
objection.

They are nonetheless silent on the same question, for a reason that has nothing to
do with estimation quality. Every one of these procedures estimates $H$, which by
\cref{prop:welded} is the self-similarity index. The kernel is a different object.
An arbitrarily clean measurement of $H$, taken from options rather than from a time
series, still carries no information about whether volatility has memory. This is
what the paper needs, and it is stronger than a complaint about proxies, because it
holds however good the proxy is.

\begin{myremark}
The Malliavin decomposition is not restricted to the fractional form. Let the variance be$V_t=\varphi(Y_t)$ with driver $Y_t=\int_0^t G(t-s)\, d W_s$, where $G$ is the response from \eqref{eq:response} and $D_u$ denotes the Malliavin derivative with respect to $W$. The chain rule yields $D_u V_t = \varphi'(Y_t) \, G(t-u)\,\mathbf{1}_{u<t}$, allowing the response to factor out of the expectation deterministically
\begin{equation*}
\EE[D_u V_t]=G(t-u)\,\EE[\varphi'(Y_t)]\,\mathbf{1}_{u<t}.
\end{equation*}
This factorization holds for any general, even non-stationary, kernel. Consequently, the limitation of current roughness estimators lies in their estimand, not in the Malliavin machinery itself. Extracting a single short-maturity exponent merely reads the diagonal of the kernel as $u\to t$. This captures the local regularity, which corresponds to the index $H$ only for a purely fractional driver. Recovering the full memory kernel is a distinct inverse problem relying on the same identity. By extracting $G$ (and subsequently the friction kernel $K$ via \eqref{eq:response}) directly from option prices, one measures the precise object estimated in \cref{sec:measure}, rather than extracting a simple roughness statistic.
\end{myremark}

The conclusion has to be drawn narrowly, because the temptation is to draw it too
widely. No estimator of $H$, from any data source, bears on the presence or the
shape of a memory kernel. Settling that is the task of the kernel measurement of
\cref{sec:measure}, whose statistical design is set out in \cref{sec:design}, and
which reads a kernel off a response function rather than an exponent off a time
series. It does not follow that volatility has memory. It follows only that the
evidence usually cited on the point speaks to the self-similarity index, which is a
different object. The memory question is decided by measuring the kernel. The
roughness estimators do not enter it.

\section{Structural freedom} \label{sec:freedom}

\subsection{A kernel with two exponents}

A fractional kernel is a single power law. It carries one exponent, and that
exponent sets both ends of the term structure at once: the local regularity of
the driver, which fixes the short-maturity SPX skew, and the decay of the kernel
at long lags, which fixes the forward variance curve. The GLE places no such constraint. The kernel may carry one exponent at short lags and another at long lags, and the two need not agree.

We parameterize the memory kernel with a short-time exponent $\beta_0$ and a tail exponent $\beta_\infty$, crossing over at a characteristic scale $1/\mu$:
\begin{equation} \label{eq:twoexp}
K(t) \sim t^{-\beta_0} \quad (t \to 0), \qquad K(t) \sim t^{-\beta_\infty} \quad (t \to \infty),
\end{equation}
which can be realized, for instance, by a tempered power law with a slow tail,
\begin{equation}
K(t) = c_0 t^{-\beta_0} e^{-\mu t} + c_\infty t^{-\beta_\infty},
\end{equation}
or by a Mittag-Leffler kernel. In the Laplace domain, the asymptotic behavior is given by $\Lap{K}(z) \sim z^{\beta_0-1}$ as $z \to \infty$ and $\Lap{K}(z) \sim z^{\beta_\infty-1}$ as $z \to 0$, and the response function $\Lap{G}(z) = (z\Lap{K}(z) + \lambda)^{-1}$ inherits both limits. Because these two limits are evaluated at opposite ends of the frequency axis, the two exponents remain entirely independent of one another.

The short-time exponent governs the short-maturity skew. At high frequency
$\Lap{G}(z)\sim z^{-\beta_0}$, so the driver spectrum decays as
$S_Y(\omega)\sim|\omega|^{-2\beta_0}$, the spectrum of a fractional Brownian
motion of index $\beta_0-\tfrac12$, and by the map preceding \cref{prop:athermal}
the short-maturity at-the-money skew scales as
$\mathcal{S}(T)\propto T^{\beta_0-1}$, \cite{ALV2007, Fukasawa2017}. The tail
exponent governs the forward variance curve. At low frequency the same response
has an algebraic tail, $G(t)\sim t^{-(1+\beta_\infty)}$, so the forward variance
relaxes to its stationary level as a power of maturity rather than exponentially,
and this slow relaxation is what VIX futures price across maturities. Raising
$\beta_0$ at fixed $\beta_\infty$ steepens the short end and leaves the
forward-variance decay in place, and raising $\beta_\infty$ at fixed $\beta_0$
does the reverse.

A fractional kernel is the degenerate case $\beta_0=\beta_\infty$. A single power law is at once the short-lag and the long-lag behavior, so the short-maturity skew and the forward variance decay become one parameter, and fitting either fixes the other. This is \Cref{prop:welded} read on the two surfaces. The two-exponent kernel breaks this tie at the level of the model structure, and the hypothesis (untested in this paper) is that this additional degree of freedom is precisely what the joint calibration requires. Whether the separation of exponents survives the passage to the risk-neutral measure, and whether it materially improves the joint fit over a single-exponent specification, are empirical questions deferred to the companion paper \cite{ItkinGLE2b}.

\medskip
The short end itself admits two readings, and the framework carries both.

\myparagraph{A rough short end}
If $\beta_0\in(\tfrac12,1)$ the kernel is singular at the origin, the driver has
the local regularity of a fractional Brownian motion of index
$\beta_0-\tfrac12\in(0,\tfrac12)$, and the short-maturity skew diverges as
$T^{\beta_0-1}$. This is the rough-volatility reading, consistent with the
realized-volatility estimates of \cite{GJR2018}, which report an index well below
$\tfrac12$. The observed value near $0.1$ corresponds to $\beta_0\approx0.6$.
Unlike a fractional model the tail exponent $\beta_\infty$ is free, so the
forward variance decay is not tied to that same $0.1$.

\myparagraph{A finite skew at zero maturity}
If instead the kernel is regular at the origin, carrying a Markovian part or a
bounded short-time piece, the driver is a semimartingale, its local regularity is
$\tfrac12$, and the short-maturity skew is finite at $T=0$ rather than divergent.
The crossover scale $1/\mu$ sets the maturity below which the skew saturates and
fixes the finite level. This is the reading of \cite{AmraniGuyon2022}, who
extrapolate a finite short-maturity skew near $1.5$ in absolute value, and of
\cite{Itkin2024}, whose ADO construction is a semimartingale of exactly this kind
and reproduces the observed skew across the traded range while remaining finite
at the origin, as \cref{prop:ado} records. Here too $\beta_\infty$ is free and
carries the memory.

\begin{myremark}
The two readings are mutually exclusive and the choice is empirical. A rough
short end predicts a short-maturity skew that grows without bound, a regular
short end predicts one that saturates, and the short end of the SPX skew term
structure decides between them. What the two share, and what a fractional kernel
denies both, is a tail exponent independent of the short end. By
\cref{sec:contdas} the memory the vanilla surface does not see is carried by
$\beta_\infty$ in either case.
\end{myremark}

\subsection{Why joint calibration wants this structure} \label{sec:joint}

The joint problem is the endpoint of this research, but is considered in detail in the companion paper \cite{ItkinGLE2b}. That is where the potential of
\cref{sec:potential} does work no kernel can do. Here, we just provide few comments on why the joint calibration problem is hard in a form that identifies which object is overloaded.

\begin{proposition}[The VIX law does not see $\rho$] \label{prop:rho}

Suppose the variance driver is autonomous, meaning that the dynamics of $Y$
involve neither $S$ nor the Brownian motion $W$ driving it. Let
\begin{equation}
\mathrm{VIX}_T^2=\frac{1}{\Delta}\,\EE_T\Bigl[\int_T^{T+\Delta}V_u\,du\Bigr].
\end{equation}
Then the law of $\mathrm{VIX}_T$ under $\mathbb{Q}$ does not depend on $\rho$.
Prices of VIX futures and VIX options therefore carry no information about
$\rho$, and $\rho$ cannot be used to fit them.
\end{proposition}

\begin{proof}
See \cref{app:rho}
\end{proof}

\cref{prop:rho} is elementary and known, but isolating it locates the
difficulty exactly. The correlation is the natural instrument for the SPX skew
and is useless on the VIX side. Whatever generates dispersion of the variance
must therefore fit the VIX smile unaided, while also supplying whatever the SPX
skew requires beyond what $\rho$ delivers. Under a quadratic potential that
entire burden falls on one vol-of-vol parameter. This is the over-constraint,
and it is a feature of the specification, not of the market - a structural limitation formally recognized since the early development of local stochastic volatility frameworks, \cite{Lipton2002}.

Forward-start options are a separate instrument for $\rho$. \cref{prop:rho} concerns the VIX and does not extend to forward-start options on the underlying.
\cite{AlosGarciaLorite2021} show that the short-maturity at-the-money level of a
Type II forward-start option is a direct function of the correlation between the
asset and its instantaneous volatility, and that the forward-start at-the-money
skew decays at a rate different from the vanilla one. Forward-start quotes
therefore constrain $\rho$ through a channel that neither the vanilla skew nor
the VIX smile provides, which is useful here and is also a caveat: the exponent
compared across start dates is not the vanilla exponent and must not be
benchmarked against it.

The literature resolves it by enlarging the state. \cite{Guyon2024} use four
path-dependent factors. \cite{AbiJaberIllandLi2022} retain a Gaussian Volterra
driver and enlarge the static link to a quintic polynomial, reporting that a
one-factor Markovian member of that class fits both surfaces well. Both add
degrees of freedom to a Gaussian core. We add them to the dynamics instead.

\subsection{The potential} \label{sec:potential}

With $U$ quadratic, $Y$ is Gaussian for every choice of $K$. The entire kernel
family of \cref{rem:white}, rough or tempered or two-exponent, varies
nothing but a covariance function. This is the class of
\cite{AbiJaberIllandLi2022}, in which non-Gaussianity is supplied afterwards by
a polynomial link applied to a Gaussian Volterra process. A non-quadratic $U$
supplies it in the dynamics instead. The two are not equivalent, and the
difference is not cosmetic.

\begin{proposition}[A link cannot reproduce a potential] \label{prop:link}

Let $\varphi$ be the quantile transform carrying the stationary law of an
Ornstein-Uhlenbeck process onto that of the overdamped double-well diffusion
$d Y=-Y(Y^2-1)d t+\sigma\,d W$. Then $\varphi(X)$ and $Y$ have identical
one-dimensional marginals, and the mean-reversion rate of $X$ may in addition be
chosen so that their integrated autocorrelation times agree, yet their regime
persistence differs.
\end{proposition}

\begin{proof}
See \cref{app:link}
\end{proof}

The reason to carry a potential, however, is not that it buys regime persistence.
It is that a barrier makes the memory kernel identifiable.

\begin{proposition}[The barrier selects the frequency at which $K$ is observed]
\label{prop:gh}
For escape over a barrier of curvature $\omega_b$ in the presence of memory, the
Grote-Hynes reactive frequency $\lambda_r$ solves
\begin{equation}   \label{eq:gh}
\lambda_r^{2} + \frac{\lambda_r}{m}\Lap{K}(\lambda_r) \;=\; \omega_b^{2},
\end{equation}
and the escape rate is $k=(\lambda_r/\omega_b)\,k_{\mathrm{TST}}$. The rate
therefore depends on $\Lap{K}$ at the single frequency $\lambda_r$, which is set
by $\omega_b$.
\end{proposition}

\begin{proof}
See \cref{app:gh}
\end{proof}

This is the strongest identification the kernel admits. The reactive frequency the
barrier isolates breaks the degeneracy that a Gaussian spectrum leaves, because it
samples $\Lap{K}$ at one point rather than through a whole covariance. The
volatility dynamics are overdamped, and there the mass term is absent from the
equation of motion. The barrier crossing then solves the first-order relation
\begin{equation} \label{eq:gh_od}
\lambda_r\,\Lap{K}(\lambda_r) = \omega_b^{2},
\end{equation}
derived in \cref{app:gh}, which fixes $\lambda_r$ from the curvature $\omega_b$ and
the friction at that frequency, \cite{GroteHynes1980, HanggiTalknerBorkovec1990}.
The reactive frequency is not the escape rate. The escape rate
$k=(\lambda_r/\omega_b)\,k_{\mathrm{TST}}$ is small because the barrier is high, and
that smallness sits in $k_{\mathrm{TST}}\sim e^{-\Delta U/\Theta}$, not in
$\lambda_r$. The reactive frequency is a property of the barrier top and the
friction, set by \eqref{eq:gh_od} alone.

The curvature $\omega_b$ is the frequency of the regime dynamics, of order the
mean-reversion rate, a few per year, and it sits well below the crossover $1/\mu$,
which is on the order of weeks. Solving \eqref{eq:gh_od} for the two-exponent kernel
of \eqref{eq:twoexp} with such a curvature and any appreciable friction places
$\lambda_r$ deep in the tail band, below the crossover by orders of magnitude, so
the regime-switching rate probes the tail exponent $\beta_\infty$. The intraday
order-flow response probes the short exponent $\beta_0$, at frequencies far above
the crossover. The two measurements land in the two bands of \eqref{eq:twoexp}, and
a fractional kernel forces them to agree because it forces $\beta_0=\beta_\infty$.
Comparing them is a test of that welding, and a separate reading of $\beta_0$ and
$\beta_\infty$ when it fails, not a fit.

\subsubsection{Which potential}

\cref{prop:gh} is the reason the potential belongs to a paper about memory. Under
a quadratic $U$ the kernel enters only through the full spectrum of a Gaussian
process, which is the degenerate case in which many kernels are observationally
equivalent over the traded range, as the Prony argument of \cref{sec:corner}
already shows. A barrier breaks the degeneracy by sampling $\Lap{K}$ at one
frequency, and varying the barrier scans the kernel. The transition rate between
volatility regimes therefore yields an estimate of $\Lap{K}$ that is independent
of the response-function estimate, and comparing the two is a test rather than a
fit.

The curvature $U''$ at the minimum is the mean-reversion rate of the variance.
Together with $\Lap{K}$ it sets the relaxation spectrum and hence the decay of
the forward variance curve, which is what VIX futures price across maturities.
The asymmetry $U'''$ governs the skewness of the stationary variance law. This
matters more than it appears, because $\rho$ does not enter the variance
equation at all and therefore cannot move the VIX smile, in this specification
or in any standard stochastic volatility model. Under a quadratic $U$ the
vol-of-vol is left to produce the SPX skew term structure and the VIX smile
simultaneously, which is the over-constraint at the centre of the joint
calibration problem, \cite{AbiJaberIllandLi2022, Guyon2024}. A nonzero $U'''$
supplies a lever that acts on the VIX side alone.

To observe this structural flexibility quantitatively, we examine the stationary distribution of the variance process $V = e^Y$, where the log-variance $Y$ is governed by the non-quadratic potential
\begin{equation} \label{potential}
U(y)=\tfrac{a}{2}y^{2}+\tfrac{c_3}{6}y^{3}+\tfrac{d}{4}y^{4},   \qquad d > 0.
\end{equation}
The stationary probability density function satisfying the Fokker-Planck equation is given by
\begin{equation}
p(y) = \frac{1}{Z} \exp\left(-\frac{2}{\sigma^2} U(y)\right), \qquad
Z = \int_{-\infty}^{\infty} \exp\left(-\frac{2}{\sigma^2} U(y)\right) \, dy,
\end{equation}
and $Z$ is the partition function. The raw moments of the transformed variance $V = e^Y$ are then computed via numerical integration over the infinite domain:
\begin{equation}
  \mathbb{E}[V^n] = \int_{-\infty}^{\infty} e^{ny} p(y) \, dy.
\end{equation}

Holding the curvature $a$ and confinement $d$ fixed preserves the overall well geometry while isolating the effect of the asymmetry parameter $c_3$. As $c_3$ is varied from $0$ to $-2$, the effect on the stationary law of $V$ is
summarized in \cref{tab:c3_skewness}. The skewness rises from $1.30$ to $1.58$ and
the mean moves from $1.15$ to $1.39$. The lever is modest, and it acts on the
variance law alone. How that asymmetry maps to the VIX implied volatility smile is a
risk-neutral computation, carried out in the companion paper \cite{ItkinGLE2b}. What
matters here is that $U'''$ moves the variance law without touching $\rho$, so it is
a degree of freedom on the VIX side that the correlation cannot supply.

\begin{table}[!htb]
\centering
\begin{tabular}{rrr}
\toprule
Asymmetry ($c_3$) & Mean ($\mathbb{E}[V]$) & Skewness ($\gamma_1$) \\
\hline
$0.0$  & $1.15$ & $1.30$ \\
$-0.5$ & $1.20$ & $1.37$ \\
$-1.0$ & $1.25$ & $1.44$ \\
$-1.5$ & $1.31$ & $1.51$ \\
$-2.0$ & $1.39$ & $1.58$ \\
\bottomrule
\end{tabular}
\caption{Moments and skewness of the stationary variance $V = e^Y$ as a function of the asymmetry parameter $c_3$ (with $a=1.0, d=1.0, \sigma=1.0$).}
\label{tab:c3_skewness}
\end{table}

The usable range of the asymmetry parameter is strictly bounded. A second stationary point of $U$ appears once $c_3^2 > 16ad$, so the bound scales as $|c_3| \le 4\sqrt{ad}$ and is a statement about the ratio, not about $c_3$ alone. At the illustrative $a=d=1$ used here the cubic overwhelms the quartic confinement beyond $c_3 \approx -2.5$. At the empirically fitted $d \approx 10^{-3}$ of \cref{tab:bistability_search} the same bound is two orders of magnitude tighter, which is why the fitted quartic places a second stationary point far outside the sampled range and why bistability must be assessed on the observed support rather than globally. The growth of $U$ at infinity controls the tail of the variance distribution and hence which moments exist, which fixes the wings of both smiles
through the moment formula, \cite{Lee2004} and bears on whether the discounted
price is a true martingale. No choice of kernel affects this, since the kernel
moves only the covariance.

A barrier, if present, contributes the regime dwell time and the frequency
$\lambda_r$ of \eqref{eq:gh}. It is worth noting that a cubic term alone is not confining, so asymmetry cannot be introduced without a quartic. The minimal potential carrying all of the above is that in \eqref{potential} with $a$ the
mean-reversion rate, $c_3$ the VIX-side asymmetry, $d$ the tail control, and
$a<0$ giving bistability and hence the barrier of \cref{prop:gh}. Three
parameters cover five roles.

The map from the physical measure to the risk-neutral measure, for a non-traded
variance driver under a non-quadratic potential, is the construction the Lean
Marketron papers, \cite{HalperinItkinLeanMarketron, HalperinItkinMarketron,
HalperinItkinMarketron2} carry out by exponential-utility indifference pricing, and
it is taken up for this model in the companion paper \cite{ItkinGLE2b}. Within the
physical law the two shape parameters are separately identified. The map from
$(c_3,d)$ to the skewness and the excess kurtosis of the stationary variance has a
non-singular Jacobian, so both are recoverable from the stationary moments, though
the two act on the moments in similar directions and the identification is only
moderately conditioned. A sharp separation uses the full smile and belongs with the
risk-neutral calibration.

Furthermore, introducing a non-quadratic $U$ breaks the affine structure of the state space, meaning the standard Riccati-Volterra machinery is lost. Pricing must therefore proceed via Monte Carlo or a numerical partial differential equation (PDE) solver operating directly in the Markovian lifted state space. For instance, the resulting anisotropic Fokker-Planck equation can be resolved using the Diagonal Frog finite-difference scheme, \cite{ItkinDF2026,Itkin2026FCDF}. This structured finite-difference approach guarantees stability over the spatial domain, allowing exact computation of the density without relying on affine characteristic functions.

The limitation of a static quantile transform, as used in polynomial diffusion models, becomes strictly evident when considering regime persistence. While a static link $\varphi$ can perfectly match the one-dimensional marginal law of a double-well potential, the marginal law is a static object whereas the dwell time is inherently dynamic. By the Kramers escape rate theory, \cite{Kramers1940}, the mean dwell time in a volatility regime separated by a barrier $\Delta U$ scales exponentially as $e^{\Delta U/\Theta}$. No such activation factor exists under any static transform of a continuous Gaussian process. Thus, a potential-driven model naturally captures the observed persistence of calm and stressed volatility regimes that a static polynomial link structurally fails to reproduce.

One caveat frames the whole measurement program. The kernel, the friction-to-noise
ratio, the barrier rate, and the leverage kernel are physical-measure quantities,
read from realized trajectories. The joint SPX and VIX calibration is a risk-neutral
exercise. The two are separate bodies of evidence, joined only by a price of risk
that must be constructed rather than read off, because the variance driver is not
traded. That construction is the subject of the companion paper \cite{ItkinGLE2b},
and none of the physical-measure tests here should be read as a statement about
prices.

\section{Markovian embedding with inertia} \label{sec:mark_embd}

In this section, we extend the analysis of the GLE to systems with finite inertia, corresponding to $m > 0$. While the dynamics in the strictly memoryless limit are well established, the presence of a non-local historical convolution introduces significant analytical and computational complexity. To address these extended memory effects, we first investigate the monotonicity properties of the inertial resolvent under completely monotone friction kernels. Subsequently, we construct an exact, higher-dimensional Markovian embedding. This lifting procedure replaces the integro-differential memory term with a coupled system of local auxiliary variables, thereby rendering the stochastic dynamics computationally tractable via standard numerical integration schemes.

\subsection{Inertia}

\begin{lemma} \label{lem:cm}
If $K$ is completely monotone and $m=0$, the resolvent is completely monotone and the forward variance impulse response is monotone, \cite{GLS1990}.
\end{lemma}

\begin{proof}
This result is established in \cite{GLS1990}. When $m=0$, the GLE reduces to a purely overdamped Volterra integral equation of the first kind. Because $K$ is completely monotone, Bernstein's theorem guarantees its Laplace transform is a Stieltjes function. The algebraic form of the overdamped response $\Lap{G}(z) = (z\Lap{K}(z) + \lambda)^{-1}$ strictly preserves this Stieltjes character, ensuring that $G(t)$ remains completely monotone and thus strictly non-negative and decreasing.
\end{proof}

\begin{proposition}[Inertia-induced oscillatory response] \label{prop:Mp}
There exists an open set of parameters $(m,\alpha,\gamma,\lambda)$ with $m>0$ for which the single-exponential friction kernel $K(t)=\gamma e^{-\alpha t}$, with $\gamma,\alpha>0$, is completely monotone yet yields a resolvent $G(t)$ that is not completely monotone and changes sign. Inertia can therefore break the complete monotonicity that \cref{lem:cm} guarantees in the overdamped case $m=0$.
\end{proposition}

\begin{proof}
See \cref{app:Mp}.
\end{proof}

\begin{myremark}
Whether the failure of complete monotonicity holds for every completely monotone $K$ (in particular, for the heavy-tailed kernels relevant to volatility modeling) remains an open question. The empirical estimates in \Cref{sec:empirical} find $O=0$, so this mathematical possibility is not active in the present data. It must nonetheless be retained as a structural caution: inertia enables dynamic behaviors that the standard overdamped limit strictly forbids, and a calibration returning $m>0$ alongside a pronounced oscillatory hump would not, on its own, indicate model misspecification.
\end{myremark}

\subsection{Lifted system for $m > 0$}

To avoid the computational burden imposed by the historical convolution in the non-Markovian friction term, we construct a Markovian embedding (or lift) for the inertial regime $m>0$. Assume the completely monotone kernel $K(t)$ is represented by a finite Prony sum of $N$ exponentials:
\begin{equation}
K(t) = \sum_{i=1}^{N} w_i e^{-\alpha_i t}, \qquad w_i > 0, \quad \alpha_i > 0.
\end{equation}

Using \eqref{eq:intro-gle}, we introduce the explicit velocity process $v_t = \dot{Y}_t$ and $N$ hidden auxiliary friction variables defined by
\begin{equation}
u_{i,t} = \int_0^t w_i e^{-\alpha_i (t-s)} v_s \, d s, \qquad i = 1, \dots, N.
\end{equation}
Differentiating $u_{i,t}$ yields a purely local evolution $d u_{i,t} = (w_i v_t - \alpha_i u_{i,t}) d t$. The non-Markovian dynamics are thus embedded exactly into an $(N+2)$-dimensional Markovian system of stochastic differential equations:
\begin{align}
d Y_t &= v_t dt, \qquad m dv_t = \left( - \sum_{i=1}^N u_{i,t} - U'(Y_t) \right) d t + \sigma dW_t, \\
d u_{i,t} &= \left(w_i v_t - \alpha_i u_{i,t}\right) d t, \qquad i = 1, \dots, N, \nonumber
\end{align}
where $W_t$ is a standard Wiener process and $\sigma$ represents the noise volatility. In transitioning to this rigorous differential form, the integrated noise increment $d\Xi_t = \sigma \, dW_t$ is introduced, such that $d\Xi_t \equiv \xi_t \, dt$.

The dimension of the lifted state space is exactly $N+2$. In the linearly restored regime where $U(y) = \frac{\lambda}{2}y^2$, the deterministic skeleton of the system is strictly linear. The eigenvalues of the corresponding drift block matrix are the roots of the generalized characteristic equation $m z^2 + z \Lap{K}(z) + \lambda = 0$. Because the weights $w_i$ and rates $\alpha_i$ are strictly positive, the friction operator dissipates energy, guaranteeing that all eigenvalues lie strictly in the left half of the complex plane. The unforced lifted system is therefore asymptotically stable.

By converting the stochastic Volterra structure into a standard, high-dimensional stochastic differential equation (SDE) system, standard numerical integrators apply directly. Under conventional regularity assumptions on the potential $U$ (such as local Lipschitz gradients and appropriate confinement), applying the Euler-Maruyama scheme to this $(N+2)$-dimensional system yields strong convergence of order $O(\Delta t^{1/2})$.

The total trajectory error is bound by the sum of this time-discretization error and the approximation error $O(\varepsilon_K)$. Here, $\varepsilon_K = \sup_{t \in [0, T]} \vert{} K(t) - K_N(t) \vert{}$ represents the maximum discrepancy over the simulation interval incurred by replacing the true completely monotone memory kernel with the discrete $N$-term Prony sum.

\section{Two-dimensional dynamics: memory in the leverage effect} \label{sec:2d}

This section sets up the two-dimensional problem and identifies which entries of the kernel matrix survive no-arbitrage.

Denote $Z_t=(X_t,Y_t)^{\top}$ with $X_t=\log S_t$, and replace \eqref{eq:intro-gle} by
\begin{equation}   \label{eq:gle2d}
\mathbf{M}\ddot Z_t \;=\; -\int_0^t\Kmat(t-s)\,\dot Z_s\,ds \;-\; \nabla U(Z_t) \;+\; \boldsymbol\xi_t,   \qquad \Kmat=
\begin{pmatrix} K_{XX} & K_{XY}\\[2pt] K_{YX} & K_{YY} \end{pmatrix}.
\end{equation}
The four entries have distinct readings. $K_{XX}$ makes the present drift of the
log-price depend on its own past increments, which is return autocorrelation.
$K_{XY}$ makes it depend on past changes in the variance driver, which is a
predictable risk premium. $K_{YX}$ makes the drift of the variance driver depend
on past returns, which is the leverage effect endowed with a term structure.
$K_{YY}$ is the kernel of the preceding sections.

\begin{proposition}[The admissible kernel matrix is lower triangular] \label{prop:triangular}
Under $\mathbb{Q}$, absence of arbitrage requires the discounted price to be a
local martingale, which forces $K_{XX}\equiv0$ and $K_{XY}\equiv0$. The entries
$K_{YX}$ and $K_{YY}$ are unconstrained by this requirement.
\end{proposition}

\begin{proof}
See \cref{app:triangular}
\end{proof}

\cref{prop:triangular} has two consequences that go beyond
bookkeeping.

\begin{corollary}[No-arbitrage forbids equilibrium] \label{cor:onsager}
An equilibrium GLE requires the memory matrix to satisfy Onsager reciprocity,
$K_{XY}=\epsilon_X\epsilon_YK_{YX}$ with $\epsilon$ the time-reversal signatures.
\cref{prop:triangular} forces $K_{XY}=0$ while the leverage effect
makes $K_{YX}\neq0$. Reciprocity therefore fails, the fluctuation-dissipation
relation \eqref{eq:fdt} cannot hold, and the market cannot be in equilibrium.
\end{corollary}

The sign of the surviving off-diagonal, and with it the direction of the
reciprocity failure, is fixed empirically by the Zumbach effect. Two points are deferred to \cref{sec:zumbach}, which makes them precise: the orientation of the inequality, and the fact that the triangularity is a $\mathbb{Q}$ statement while the Zumbach asymmetry is measured under $\mathbb{P}$, so identifying the two off-diagonals across the measures rests on a standing assumption on the variance risk premium stated there. The upgrade to \cref{sec:measure} is then substantive: the failure of fluctuation-dissipation stops being an assumption that markets are out of equilibrium and becomes a consequence of no-arbitrage together with a measured asymmetry.

\myparagraph{Correlation $\bm\rho$ is the memoryless limit of $\bm K_{YX}$} \label{rem:rho-limit}
Setting $K_{YX}=\rho\,\delta$ recovers instantaneous price-variance correlation
and nothing more. A kernel with support gives the leverage effect a term
structure, which $\rho$ cannot represent. \cite{BMP2001} report that the
empirical return-volatility correlation decays over a characteristic time of
approximately 10 days for equity indices (specifically 9.3 days) and considerably longer for single names (averaging 50 days, with US equities near 69 days and European and Japanese equities near 40 days),
which is a direct measurement of $K_{YX}$ and is incompatible with a delta.

\myparagraph{Path-dependent volatility as a restricted GLE.} \label{rem:pdv}
The four-factor path-dependent volatility model of \cite{Guyon2024} makes volatility a deterministic function of exponentially weighted averages of past returns and past squared returns, at two timescales each. The memory architecture of that model corresponds to the two-dimensional GLE \eqref{eq:gle2d} with a specific set of choices: the off-diagonal kernel $K_{YX}$ and the variance-memory kernel $K_{YY}$ are each truncated to two exponential modes, the potential is quadratic, the noise is absent from the variance driver (making the path deterministic given returns), and the link function $\varphi$ is of affine-plus-square-root form.

Under these restrictions the mapping can be made explicit. Write the path-dependent model as
\begin{equation} \label{eq:gl-model}
  \sigma_t = \beta_0 + \beta_1 R_{1,t} + \beta_2\sqrt{R_{2,t}},
  \qquad
  R_{k,t} = \sum_{i=0}^{1}\omega_{k,i}\!\int_{-\infty}^{t}\!\lambda_{k,i}\,
  e^{-\lambda_{k,i}(t-s)}\Big(\tfrac{dS_s}{S_s}\Big)^{k},
\end{equation}
with $k=1$ for returns and $k=2$ for squared returns. Take the second-row kernels of \eqref{eq:gle2d} as two-mode sums, $K_{YX}(\tau)=\sum_i g_i e^{-\eta_i\tau}$ and $K_{YY}(\tau)=\sum_i c_i e^{-\nu_i\tau}$, and lift each memory integral to auxiliary states,
\begin{equation} \label{eq:pdv-lift}
dh_i = -\eta_i h_i\,dt + g_i\,dX, \quad
dp_i = -\nu_i p_i\,dt + c_i\,dY, \quad
\int_0^t\! K_{YX}\,dX = \sum_i h_i, \quad
\int_0^t\! K_{YY}\,dY = \sum_i p_i.
\end{equation}
The leverage states $h_i$ are exponentially weighted averages of past returns, so they coincide exactly with the return factors $R_{1,i}$ of \eqref{eq:gl-model}, with $\eta_i=\lambda_{1,i}$ and $g_i$ set by $\omega_{1,i}$ and $\beta_1$. The variance-memory states $p_i$ are exponentially weighted averages of past increments of the driver $Y$, whereas $R_{2,i}$ averages past squared returns. The two coincide only when the variance is read deterministically off the price path (that is, when $dY\propto(dS/S)^2$ and the noise $\boldsymbol\xi$ in \eqref{eq:gle2d} is switched off) and when the link $\sqrt{\varphi(\cdot)}$ is matched to the affine-plus-root form of \eqref{eq:gl-model}.

Away from this restricted regime the two models part. The gap is precisely the three generalizations the GLE framework supplies: an independent variance innovation, a general link function $\varphi$, and kernels richer than two exponentials. The model of \cite{Guyon2024} is therefore recovered as the intersection of the GLE class with a set of strong parametric restrictions, not as a special case in the sense of nested parameter spaces. Because the noise is absent from the deterministic-path regime, the fluctuation-dissipation violation of \Cref{cor:onsager} cannot even be formulated there; restoring the noise makes it testable.

\subsection{The Zumbach effect fixes the sign}  \label{sec:zumbach}

The asymmetry of \cref{cor:onsager} is not a formal curiosity. It has been
measured, it has a name, and its direction is known, which settles the sign that
\cref{cor:onsager} leaves open.

\cite{ZumbachLynch2001} observed that financial time series are not statistically
symmetric under exchange of past and future, an effect since documented across
markets and centuries of data, \cite{ChicheporticheBouchaud2014}. In the form
relevant here it reads
\begin{equation}
  \EE^{\mathbb{P}}\!\left[r_t^2\,\sigma_{t+\tau}^2\right]
  \;>\;
  \EE^{\mathbb{P}}\!\left[r_{t+\tau}^2\,\sigma_t^2\right],
  \qquad \tau>0,
  \label{eq:zumbach}
\end{equation}
where $r_t$ is the return over a short interval at $t$ and $\sigma_t$ the
volatility, and the expectations are under the physical measure $\mathbb{P}$
because \eqref{eq:zumbach} is estimated from realized time series. Past squared
returns forecast future volatility more strongly than past volatility forecasts
future squared returns. Equivalently, and in the words
of the microstructural literature, past returns affect future volatility but not
the other way around, \cite{DandapaniJusselinRosenbaum2021}.

In the language of \eqref{eq:gle2d} the two sides of \eqref{eq:zumbach} are the
two off-diagonal couplings. The left side is the action of past returns on the
present variance, carried by $K_{YX}$. The right side is the action of past
variance on present returns, carried by $K_{XY}$. The inequality
\eqref{eq:zumbach} therefore states precisely that $K_{YX}$ dominates $K_{XY}$,
and \cref{prop:triangular} sends the weaker of the two to zero exactly. The
Zumbach inequality and the no-arbitrage triangularity are the same statement
about the same matrix, one measured and one derived.

Which measure, and why it matters here? Two measures meet in this section and the distinction is not cosmetic. The dynamics \eqref{eq:gle2d},
\cref{prop:triangular} and \cref{cor:onsager} are stated under the pricing
measure $\mathbb{Q}$, since the martingale condition that produces the
triangularity is the $\mathbb{Q}$ local-martingale property of the discounted
price. The Zumbach inequality \eqref{eq:zumbach} is a property of realized time
series and therefore holds under the physical measure $\mathbb{P}$. Equating its
two sides with $K_{YX}$ and $K_{XY}$ requires those kernels to be the same under
both measures, which is not automatic.

It is, however, defensible, and precisely for the off-diagonal entries. The
change from $\mathbb{P}$ to $\mathbb{Q}$ by Girsanov shifts drifts and leaves
quadratic variation invariant. The off-diagonal kernels enter the drifts of $Z$,
so a general market price of risk can rescale them and can, in principle, reverse
the ordering in \eqref{eq:zumbach}. We therefore take as a standing assumption
that the variance risk premium is not so large as to reverse the sign of the
dominant off-diagonal, under which the $\mathbb{P}$ asymmetry transfers to
$\mathbb{Q}$ with its orientation intact. This is the market price of memory risk
named in the discussion following \cref{prop:triangular}, and it is the honest
location of the $\mathbb{P}$-to-$\mathbb{Q}$ gap.

The two-dimensional dynamics and \cref{prop:triangular,cor:onsager} are most
naturally read under $\mathbb{P}$, since the leverage and Zumbach effects are
physical-measure phenomena, while pricing occurs under $\mathbb{Q}$. The
construction that connects the two, and the reason the connection is not a mere
change of drift, is the subject of \cite{ItkinGLE2b}.

This does three things for the paper. First, it removes the ambiguity flagged
after \cref{cor:onsager}. The direction of the inequality fixes which
off-diagonal survives, so the time-reversal signatures need not be argued from
first principles - they are read off \eqref{eq:zumbach}. The market's failure of
fluctuation-dissipation is oriented, and it is oriented the way the data says.

Second, it identifies our $K_{YX}$ with a quantity that has an independent
measurement history. The Zumbach effect is conventionally reproduced by quadratic
Hawkes processes, \cite{BlancDoncBouchaud2017}, whose scaling limit is a
super-Heston rough volatility model, \cite{DandapaniJusselinRosenbaum2021}. That
is the same microstructural lineage that produces rough Heston in the linear
case, now carrying the off-diagonal memory explicitly, and it gives \cref{T7} a target
whose sign and rough order of magnitude are already known.

Third, it sharpens the distinction from a purely diagonal model. A model with
$K_{XY}=K_{YX}=0$, or with the memoryless $K_{YX}=\rho\delta$ of
\cref{rem:rho-limit}, cannot produce \eqref{eq:zumbach}, because a symmetric or
instantaneous coupling makes the two sides equal. The Zumbach effect is thus
direct evidence for the off-diagonal kernel with support, and against the two
specifications the paper is arguing past.

\begin{myremark}[Weak and strong forms] \label{rem:zumbach-strong}
The inequality \eqref{eq:zumbach} is the weak Zumbach effect, and it constrains
$K_{YX}$. A strong form also holds empirically: the conditional law of future
volatility depends on the past volatility path and is not a function of the current
instantaneous variance alone, \cite{Guyon2024}. The extra dependence is on the
history of the volatility, and it enters through $K_{YY}$, which must therefore
carry memory rather than act instantaneously. This is the non-Markovianity the
kernel of the preceding sections supplies.

The strong form does not force $K_{YY}$ to be non-completely-monotone. Complete
monotonicity and genuine path-dependence are separate properties. Complete
monotonicity fixes the sign and shape of the kernel, a positive mixture of decaying
exponentials, and by \cref{lem:cm} it makes the forward-variance impulse response
monotone. The strong form fixes only that the kernel carries memory. A completely
monotone kernel with support, the power law of the rough regime among them, is
already a genuine functional of the path, so it satisfies the strong form while
keeping the monotone response of \cref{lem:cm}. Non-complete-monotonicity would be
required only if volatility responded to a shock by overshooting rather than
relaxing, and neither the strong Zumbach effect nor the measured persistence of
volatility asks for that. The strong form therefore sits alongside \cref{lem:cm} and
the inertia discussion, constraining the support of $K_{YY}$ and not its complete
monotonicity.
\end{myremark}

The signatures follow from what the variables are. Under $t\to-t$ the log-price
$X$ is a configurational coordinate and is even, so its increment $r$ is odd and
$r^2$ even. The log-variance driver $Y$ is likewise configurational and even, and
$\sigma^2$ is even. Thus $\epsilon_X=\epsilon_Y=+1$. The signatures are assigned to the configurational coordinates $X$ and $Y$ themselves; their velocities $\dot X$ and $\dot Y$ are the odd, Casimir-type variables and, in the inertial case $m>0$, carry the opposite signature without affecting the reciprocity relation used here. Onsager-Casimir reciprocity
for the memory matrix reads $K_{ij}(\tau)=\epsilon_i\epsilon_j K_{ji}(\tau)$, so
with both signatures positive an equilibrium market would carry $K_{XY}=K_{YX}$.
Two equal off-diagonals give a symmetric Zumbach correlation, $\EE^{\mathbb
P}[r_t^2\sigma_{t+\tau}^2]=\EE^{\mathbb P}[r_{t+\tau}^2\sigma_t^2]$, since $r^2$
and $\sigma^2$ are both even. The measured inequality \eqref{eq:zumbach} is
precisely the failure of that symmetry. It forces $K_{XY}\neq K_{YX}$, and its
orientation, with $K_{YX}$ the larger, fixes which entry survives the
triangularity of \cref{prop:triangular}. The reciprocity computation and the
empirical anchor return the same answer, and \cref{cor:onsager} is closed: the
surviving off-diagonal is $K_{YX}$, and the reciprocity failure is oriented as
\eqref{eq:zumbach} orients it.

The numerical results are obtained from the Markovian lift of the leverage-kernel model, specified as
\begin{align*}
dX &= \sigma\,dW^X,\qquad \sigma=e^{Y/2},\qquad dh=-\eta h\,dt+g\,dX, \\
dY & = [-\kappa(Y-\mu_Y)-h]\,dt+\nu\,dW^Y, \qquad  \langle dW^X,dW^Y\rangle=\rho\,dt,
\end{align*}
with parameters $\kappa=3$, $\nu=1$, $\eta=25$, and $e^{\mu_Y}=0.04$. The leverage kernel takes the exponential form $K_{YX}(\tau)=g e^{-\eta\tau}$, and the table below reports the percentage changes in the stationary mean and standard deviation of $V$ as $\rho$ moves from $0$ to $-0.9$ for increasing values of the kernel amplitude $g$.

The leverage kernel leaves the VIX nearly $\rho$-blind, breaking the autonomy hypothesis of Proposition~\ref{prop:rho} by introducing a nonzero $K_{YX}$. This coupling makes the drift of $Y$ depend on past returns, and hence on the price and its driving Brownian motion, so that the variance law no longer remains independent of $\rho$. The dependence is linear in the leverage kernel. The coupling injects $\int K_{YX}\,dX$ into $Y$, a state whose correlation with the variance noise enters solely through $\rho$. Consequently, the shift in the variance law appears at order $\rho K_{YX}$, while its sensitivity to $\rho$ is of order $K_{YX}$. These scalings, together with their magnitudes, are confirmed on the Markovian lift in Table~\ref{tab:rho-leak}. When $K_{YX}=0$, the driver is autonomous and the forward variance is exactly independent of $\rho$, as required by the proposition. Once $K_{YX}$ is activated, however, varying $\rho$ over its full range shifts the forward variance level by less than one percent even under a strong leverage kernel, and alters the variance dispersion by a few percent. Both effects scale linearly with the kernel amplitude. Thus, the VIX futures level is effectively $\rho$-blind, and the VIX smile carries only a small $\rho$-sensitivity that scales with the leverage memory. The instrument separation established in Section~\ref{sec:joint} therefore survives up to a correction of this order, and $\rho$ remains an SPX-side instrument.

\begin{table}[!htb]
\centering
\begin{tabular}{crr}
\toprule
$g$ & $\Delta\,\mathbb{E}[V]$ & $\Delta\,\mathrm{sd}[V]$ \\
\hline
$0$ & $0.00\%$ & $0.00\%$ \\
$1$ & $0.12\%$ & $0.96\%$ \\
$3$ & $0.35\%$ & $2.88\%$ \\
$6$ & $0.71\%$ & $5.77\%$ \\
\bottomrule
\end{tabular}
\caption{Sensitivity of the stationary variance law to $\rho$ as the leverage-kernel amplitude $g$ increases. Columns give percentage changes in the mean and standard deviation of $V$ when $\rho$ varies from $0$ to $-0.9$. The effect vanishes at $g=0$, as required by Proposition~\ref{prop:rho}.}
\label{tab:rho-leak}
\end{table}

\section{Various tests and falsification criteria} \label{sec:design}

To validate the theoretical architecture developed in the preceding sections, this section presents a comprehensive empirical assessment of the generalized Langevin framework. The evaluation systematically examines each core component, including Bayesian estimation of memory kernels, fluctuation-dissipation consistency, forward-variance inertial effects, macro-scale barrier-crossing dynamics, and multi-dimensional leverage structures. For each component, we define explicit quantitative benchmarks designed to detect deviations from model predictions and to delineate the boundaries of model applicability.

We establish the following definitive kill criteria based on the preceding test sequence:
\begin{itemize}
\item \emph{Absence of Memory (\cref{T1}):} If the estimation yields an exponential kernel, the volatility process possesses no long-range memory. The framework strictly reduces to a lifted Heston model with redundant parameters.
\item \emph{Rough Scaling Constraint (\cref{T2}):} If the empirical data cannot reject the constraint $\beta_0=\beta_\infty$, the decoupled multi-band freedom claimed by the model is unutilized, vindicating the single-parameter rough volatility paradigm.
\item \emph{Equilibrium Dynamics (\cref{T3}):} If the fluctuation-dissipation ratio $\Xfdt(\omega)$ is statistically indistinguishable from unity, the driver noise is purely thermal, rendering the athermal extensions in \cref{sec:measure} physically vacuous.
\item \emph{Absence of Inertia (\cref{T4}):} If the forward-variance impulse response is monotonically non-increasing within the noise limits, the variance process lacks inertia, forcing the model to set $m=0$.
\item \emph{Kernel Inconsistency (Barrier Crossing):} If the macroscopic regime-transition rate disagrees with the high-frequency kernel estimation of \cref{T1} outside their joint credible intervals, the memory kernel is structurally unidentifiable, requiring the withdrawal of \cref{sec:potential}.
\item \emph{Instantaneous Leverage (\cref{T6}):} If the empirical leverage term structure cannot reject a memoryless delta function ($K_{YX} = \rho\,\delta$), the leverage dynamics reduce to a standard one-dimensional scalar correlation, invalidating the two-dimensional memory formulation of \cref{sec:2d}.
\end{itemize}

Each of these empirical outcomes is explicitly reported in our empirical results below. The ultimate risk-neutral kill criterion --- that the framework's joint option surface fit must statistically outperform a memoryless one-factor model at an equivalent parameter count, is formally stated and tested in the companion paper \cite{ItkinGLE2b}.

\subsection{Bayesian Estimation of the Kernel and the Limits of the Rough Test} \label{test:t1}

\Cref{sec:measure} established that the memory kernel is theoretically recoverable from a response function when the data-generating process is known. The present section asks a different question: given a finite sample of market data and a kernel with two free exponents, which features of the kernel does the likelihood actually identify, and which remain prior-driven? The distinction between recoverability in principle and identification in practice is an important empirical finding of this paper.

We probe volatility memory using estimation tests: \test{}\label{T1} and \test{}\label{T2} on high-frequency limit order book data from the FI-2010 benchmark dataset, \cite{ntakaris2018benchmark}. We extract continuous-trading records from No-Auction cross-validation fold 1 (\texttt{Train\_Dst\_NoAuction\_ZScore\_CF\_1.txt}) to build the empirical time grid. This yields Level-1 order flow imbalance shocks ($\xi_t$) and realized variance responses ($Y_t$). The sample contains $N=15{,}000$ observations using a 200-event rolling window for variance estimation.

The regression estimates the response rather than the friction, which dictates the parameterization. Letting $G$ be the fitted kernel, the model is $Y=G*\xi$ with $\Lap{G}(z)=(z\Lap{K}(z)+\lambda)^{-1}$ for the friction kernel $K$ of \eqref{eq:twoexp}. A friction $K(t)\sim t^{-\beta}$ yields $\Lap{K}(z)\sim z^{\beta-1}$. This produces $\Lap{G}(z)\sim z^{-\beta}$ and $G(t)\sim t^{\beta-1}$ wherever $z\Lap{K}(z)$ dominates $\lambda$. We sample the friction exponents directly and specify the response model
\begin{equation} \label{eq:t1model}
G(t) = c_0\,t^{\beta_0-1}e^{-\mu t} + c_\infty\,t^{\beta_\infty-1}.
\end{equation}
A No-U-Turn (NUTS) Hamiltonian Monte Carlo sampler estimates the parameters with the convolution truncated at $L=100$ lags. The estimation window dictates the second exponent. Once $\lambda$ dominates, the response tail is the $t^{-(1+\beta_\infty)}$ of \cref{sec:freedom} and decays faster than $t^{-1}$. The fitted exponents fall well below unity. The window thus sits where $\lambda$ remains negligible, validating \eqref{eq:t1model} at both boundaries. The asymptotic tail lies beyond this window and is unestimated.

The exponent priors must satisfy three requirements underlying the test. They require support on $(0,1)$ to exclude inadmissible kernel values. They are anchored via the model relation $H=\beta-\tfrac12$, placing the short end in the rough range. They also share identical marginals to keep the \cref{T2} null hypothesis $\beta_0=\beta_\infty$ a priori available. We specify $\beta_0,\beta_\infty\sim\mathrm{Beta}(5.83,2.50)$ with mean $0.70$ and standard deviation $0.15$. This places $89\%$ of the implied $H_0$ inside $(0,\tfrac12)$ and assigns $19\%$ of prior mass to $|\beta_0-\beta_\infty|<0.05$. We omit ordering constraints because the test aims to measure the separation direction. The amplitudes $c_0$ and $c_\infty$ and the cutoff $\mu$ take half-normal priors with scales $5$, $5$, and $0.5$.

The sampler converged reliably. Tuning left two divergences, $\hat{R}$ equals $1.00$ for all parameters, and effective sample sizes exceed $3{,}600$. These diagnostics confirm thorough posterior exploration. They do not indicate whether the data or the prior drives this posterior. We must answer that question to draw valid conclusions. \cref{tab:t1_results} therefore reports the prior standard deviation and the shrinkage ratio (posterior to prior standard deviation) for each parameter.

\begin{table}[!htb]
\centering
\renewcommand{\arraystretch}{1.2}
\begin{tabular}{lccccc}
\hline
\textbf{Parameter} & \textbf{Mean} & \textbf{SD} & \textbf{89\% ETI} & \begin{tabular}{@{}c@{}}\textbf{Prior} \\ \textbf{SD}\end{tabular} & \begin{tabular}{@{}c@{}}\textbf{Shrink-} \\ \textbf{age}\end{tabular} \\
\hline
\multicolumn{6}{c}{\textit{Friction exponents}} \\
\hline
$\beta_0$ (short-time) & 0.709 & 0.148 & $[0.450,\,0.920]$ & 0.150 & 0.99 \\
$\beta_\infty$ (long-time) & 0.799 & 0.111 & $[0.590,\,0.940]$ & 0.150 & 0.74 \\
\hline
\multicolumn{6}{c}{\textit{Response amplitudes and crossover}} \\
\hline
$\mu$ (cutoff) & 0.462 & 0.361 & $[0.016,\,1.100]$ & 0.301 & 1.20 \\
$c_0$ (short amplitude) & 0.0082 & 0.0080 & $[0.0005,\,0.023]$ & 3.014 & 0.003 \\
$c_\infty$ (long amplitude) & 0.0143 & 0.0065 & $[0.0037,\,0.025]$ & 3.014 & 0.002 \\
\hline
\multicolumn{6}{c}{\textit{Implied local regularity, short end only}} \\
\hline
$H_0=\beta_0-\tfrac12$ & 0.209 & 0.148 & $[-0.050,\,0.420]$ & 0.150 & 0.99 \\
\hline
\end{tabular}
\caption{Posterior summaries for the \cref{T1} estimation on FI-2010 data under admissible, same-marginal exponent priors. Shrinkage is the ratio of posterior to prior standard deviation. A value near unity indicates a parameter uninformed by the likelihood. Only the short end carries a Hurst index. The long-time exponent $\beta_\infty$ governs forward-variance decay rather than local regularity.}
\label{tab:t1_results}
\end{table}

The shrinkage metric groups the parameters into three categories. The amplitudes are sharply identified. Their standard deviations drop more than two orders of magnitude below prior values, and their means shift from roughly $4$ to near $0.01$. The tail exponent is also identified. Its distribution tightens by a quarter while its mean moves from $0.700$ to $0.799$. The short-time exponent and the crossover remain unidentified. The posterior standard deviation of $\beta_0$ is $0.99$ times its prior value, and that of $\mu$ exceeds its own prior. The likelihood adds no information to either parameter. The implied local regularity inherits this non-identification. Its posterior mean $H_0=0.209$ matches the prior mean $0.200$, and its interval still admits inadmissible negative values.

Varying the prior confirms this reading and provides the clearest evidence of what the data determine. \cref{tab:prior_sensitivity} reports the friction exponents under three prior specifications with different centres, widths, and supports. The tail exponent stabilizes near $0.80$ under both priors that permit likelihood influence. It agrees to three decimals across priors that otherwise contradict each other. The short-time exponent simply reproduces its assigned prior mean. The tail exponent is thus a measurement. The short-time exponent is an assumption. This robustness is across priors on a single dataset. The measured level is not invariant across datasets, as the estimates below show, so $\beta_\infty$ is determined by the data in each case while its value reflects the sampling timescale and the variance window.

\begin{table}[!htb]
\centering
\renewcommand{\arraystretch}{1.2}
\begin{tabular}{lcccc}
\hline
\textbf{Prior on the exponents} & $\beta_0$ & $\beta_\infty$ & $\Delta\beta$ 95\% HDI & \textbf{\cref{T2} verdict} \\
\hline
Informative, centres $0.70$/$0.25$, SD $0.08$/$0.10$ & 0.709 & 0.274 & $[\phantom{-}0.158,\,0.684]$ & rejects \\
Uniform on $[0.01,0.99]$ & 0.515 & 0.798 & $[-0.945,\,0.329]$ & fails to reject \\
Admissible, $\mathrm{Beta}(5.83,2.50)$ on both & 0.709 & 0.799 & $[-0.473,\,0.266]$ & fails to reject \\
\hline
\end{tabular}
\caption{Posterior means for the friction exponents and the \cref{T2} interval across three prior specifications. The tail exponent is robust whenever the prior permits. The short-time exponent tracks the prior mean and remains undetermined by the data. Only the first specification rejects the rough constraint. Its prior assigns just $0.2\%$ of mass to the tested null, and its centre for $\beta_\infty$ lies $5.5$ prior standard deviations from the consensus of the other runs.}
\label{tab:prior_sensitivity}
\end{table}

\Cref{T2} provides the falsification criterion for the rough constraint $\beta_0=\beta_\infty$. The admissible prior yields a posterior difference $\Delta\beta=\beta_0-\beta_\infty=-0.090$ with a 95\% HDI of $[-0.473,\,0.266]$. The null hypothesis avoids rejection because zero falls within this interval. \cref{fig:t2_posterior} compares this posterior against the other two options to illustrate the central argument. Only the informative specification produces a posterior excluding zero. That prior assigns just $0.2\%$ of its mass to the tested null neighborhood. The uniform and admissible priors assign $9.9\%$ and $19\%$ respectively. Both leave the null within the credible region. A parameter separation that appears, vanishes, or reverses sign depending on the prior reflects assumptions rather than empirical data.

This prior dependence ensures the non-rejection is genuine. The null hypothesis was truly open to falsification. However, non-rejection does not prove equality. The interval remains wide because an underlying parameter is unconstrained. A test relying on prior information for a parameter cannot discriminate between hypotheses about it. Following \cref{test:fdt}, a failure to reject is a definitive result only with established statistical power against the alternative. That power is absent here.

\begin{figure}[!htp]
\centering
\includegraphics[width=0.7\textwidth]{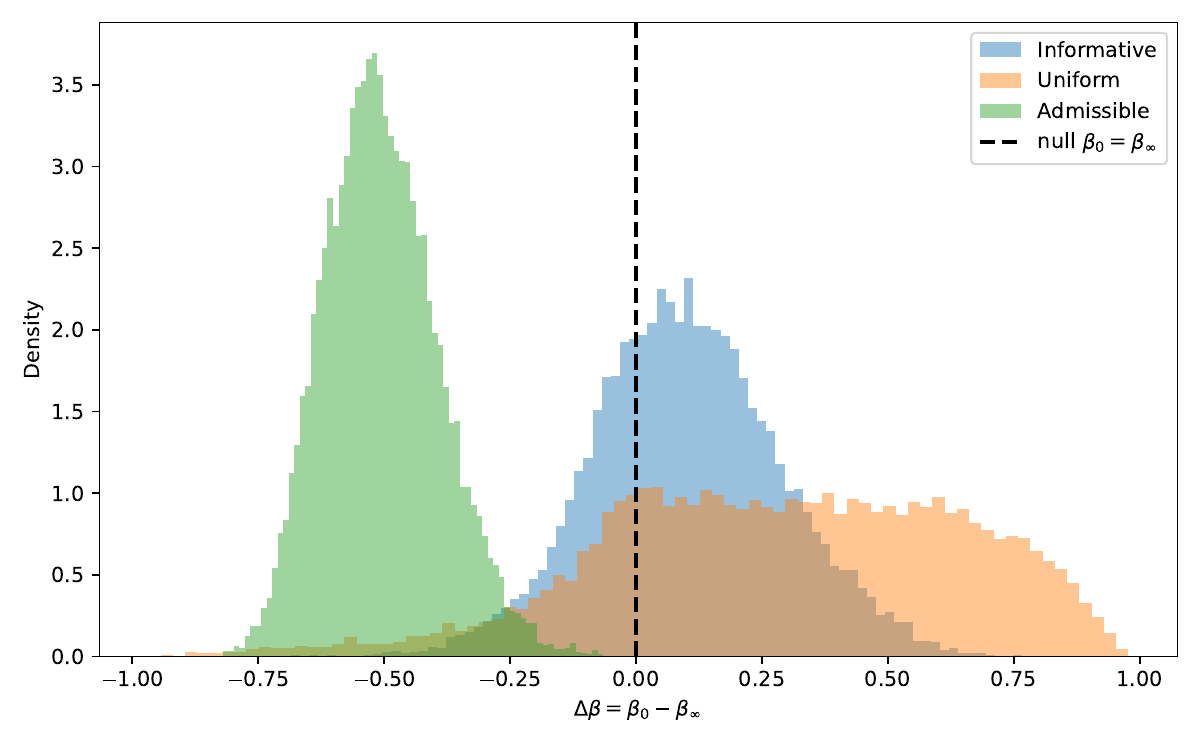}
\caption{Posterior distribution of the friction exponent difference $\Delta\beta=\beta_0-\beta_\infty$ (\cref{T2}) on FI-2010 data under admissible, same-marginal exponent priors. The $95\%$ Highest Density Interval is shaded and contains the zero reference line (dashed red), preventing rejection of the rough constraint. The interval width constitutes the substantive content, reflecting the non-identification of $\beta_0$ from \cref{tab:t1_results}.}
\label{fig:t2_posterior}
\end{figure}

The non-identification of $\beta_0$ is not a failure of the estimation procedure, it is the empirical obstacle anticipated in \Cref{sec:measure}. The crossover scale $1/\mu$ that separates the two power-law regimes is not resolved within the available lag window, so the short-time exponent and the crossover rate trade off against one another with negligible change in fit. This trade-off is a feature of the finite sample, not of the estimator: the likelihood is flat along a ridge in $(\beta_0,\mu)$, and the prior supplies the curvature that the data cannot. The result is that the tail exponent $\beta_\infty$ is measured, while the short-time exponent $\beta_0$ remains an assumption—exactly the pattern that Section~4.1 warned could arise when the two-exponent structure is estimated on a single band.

The verdict is therefore mixed. The data successfully measure the order-flow response amplitude and the friction kernel tail exponent. This latter measurement constitutes an important finding. The stable value $\beta_\infty\approx0.80$ indicates slow algebraic decay. This supplies the long-memory component required by the \Cref{sec:freedom} framework.

The data fail to measure the short-time exponent. They consequently reveal nothing about the local regularity of the driver. The estimation neither supports nor excludes cases (A) or (B) of \Cref{sec:freedom}. The two-exponent structure and the rough constraint remain neither confirmed nor falsified. A shape degeneracy in \eqref{eq:t1model} causes this obstruction. The short exponent trades against the crossover $\mu$ over the hundred-lag window with almost no change in fit. This trade-off leaves both parameters unidentified while clearly resolving the tail and amplitudes.

To verify the FI-2010 results on a substantially richer dataset, the estimation
was repeated on Monthly TAQ NBBO data for AAPL (January 15, 2025, 09:30--10:00),
obtained from WRDS. The NBBO file provides every consolidated quote update at
millisecond resolution with posted bid and ask volumes directly observable,
allowing the same order flow construction used for FI-2010. The sample contains
72{,}058 quotes at an average rate of 40 per second, resampled to a 20\,ms grid
(89{,}997 observations). Under the lag-dependent exponent model,
$\beta_\infty$ is sharply identified at $0.871$ (shrinkage $0.31$, 89\% ETI
$[0.80, 0.94]$), a stable measurement of the tail exponent that governs the
forward-variance decay. The
short-time exponent $\beta_0$ remains completely prior-driven (shrinkage
$1.03$), and its implied Hurst index $H_0 = \beta_0 - 1/2$ inherits this
non-identification (posterior mean $0.129$, 89\% ETI $[-0.21, 0.41]$). The
rough constraint is not rejected (95\% HDI for $\beta_0-\beta_\infty$:
$[-0.577, 0.068]$). This confirms the FI-2010 finding on the highest-quality
data available: the tail exponent of the memory kernel is identified, while
the short-time exponent and the local regularity of the driver are not, at any
resolution achievable with current limit order book data. What is stable is the
identification and not the level, since the tail exponent reads $0.87$ here
against $0.80$ on FI-2010 and tracks the sampling scheme even where the data
determine it.

The estimation was also attempted on several other datasets to assess whether
higher resolution or a different asset class could break the degeneracy between
$\beta_0$ and $\mu$. These included TrueFX EUR/USD tick data at ten millisecond
resolution, Kraken BTC/USDT trade data at one hundred millisecond resolution,
and one-minute Bitcoin intraday bars from Yahoo Finance resampled to grids
ranging from one to sixty seconds. In every case, both exponents remained
unidentified (shrinkage above $0.95$ for all parameters). Inspection of the
order flow proxies used in these auxiliary datasets suggests the cause: outside
the limit order book, where $\xi_t$ is constructed directly from changes in
posted bid and ask volumes, the available proxies (signed tick direction, trade
side imbalance, price direction times volume) are too weakly correlated with the
latent driving noise for the response kernel to be recoverable. The FI-2010
event-time results therefore remain the cleanest physical-measure estimates
available. They establish that $\beta_\infty$ is measured and $\beta_0$ is not,
and they indicate that resolving the short-time exponent requires either richer
order book data or the risk-neutral approach of the companion paper.

A decisive test requires a longer lag range rather than more observations at the current scale. Breaking the trade-off requires separating the two bands by multiple decades. A power calculation must also fix this required span in advance. We record this outcome instead of the falsification initially intended. The risk-neutral route of \cite{ItkinGLE2b} avoids this degeneracy entirely, because distinct instrument maturities probe the two bands independently.

\subsection{The fluctuation-dissipation test} \label{test:fdt}

\Test{}\label{T3} evaluates market equilibrium by comparing the measured driver spectrum against the theoretical dissipation requirement. In Lamperti coordinates, the working fluctuation-dissipation ratio is
\begin{equation} \label{eq:Xfdt-working}
\Xfdt(\omega)\;\propto\;\frac{\mathrm{Re}\,\Lap{K}(i\omega)}{S_\xi(\omega)}.
\end{equation}
Under the equilibrium null $H_0$, $S_\xi = \Theta\,\mathrm{Re}\,\Lap{K}$ renders $\Xfdt(\omega)$ constant across frequencies. Under the athermal alternative (\cref{prop:athermal}), flat noise yields $\Xfdt(\omega)\propto\omega^{\alpha-1}$. Because $\Theta$ is identified only up to scale, every statistic depends solely on the shape of $\log\Xfdt$ against $\log\omega$.

The ratio is estimated from binned periodogram data $\widehat{S}_Y(\omega)$ via
\begin{equation} \label{eq:Xfdt-est}
\widehat{\Xfdt}(\omega) \propto \frac{\mathrm{Re}\,\Lap{K}(i\omega)\,\bigl|\tilde R(\omega)\bigr|^{2}} {\widehat{S}_Y(\omega)}, \quad \tilde R(\omega) = \bigl(\lambda + i\omega\Lap{K}(i\omega)\bigr)^{-1}.
\end{equation}
Log-averaging within frequency bins tames periodogram tail variance. We construct two scale-invariant statistics evaluated against surrogate null distributions: the OLS slope $\hat{s}$ (directional, targeting $\alpha-1$) and the log-spread $\hat{\sigma}$ (omnibus flatness).

Power is benchmarked on \cref{cor:discriminator} ($H=0.1$ for both corners: $\alpha=0.2$ for equilibrium, $\alpha=0.6$ for athermal). As shown in \cref{tab:fdt_power}, the directional slope dominates the omnibus spread, reaching complete power at $N=4{,}096$ with an unbiased slope estimate ($\hat{s} = -0.40$).

\begin{table}[!htp]
\centering
\renewcommand{\arraystretch}{1.2}
\begin{tabular}{rccc}
\hline
$N$ & \begin{tabular}{@{}c@{}}\textbf{SD of} $\hat{s}$ \\ \textbf{under} $H_0$\end{tabular} & \begin{tabular}{@{}c@{}}\textbf{Power,} \\ \textbf{slope} $\hat{s}$\end{tabular} & \begin{tabular}{@{}c@{}}\textbf{Power,} \\ \textbf{spread} $\hat{\sigma}$\end{tabular} \\
\hline
1,024 & 0.233 & 0.48 & 0.10 \\
4,096 & 0.089 & 1.00 & 0.72 \\
16,384 & 0.046 & 1.00 & 1.00 \\
65,536 & 0.025 & 1.00 & 1.00 \\
\hline
\end{tabular}
\caption{Power of \cref{T3} against the athermal alternative of \cref{cor:discriminator} (5\% level, 600 replications). The directional slope achieves complete separation ($\hat{s} = -0.400$) by $N=4{,}096$.}
\label{tab:fdt_power}
\end{table}

Sampling error is not the limiting factor but kernel precision is. An exponent error $\delta$ in \eqref{eq:Xfdt-est} induces $\widehat{\Xfdt} \propto \omega^{-\delta}$, biasing the measured slope by exactly $-\delta$. Simulation confirms this 1:1 linear bias across admissible errors: under the null, $\delta \in \{-0.19, -0.10, 0.10,
0.20, 0.39\}$ yields measured slopes of $+0.19, +0.10, -0.09, -0.20,$ and
$-0.39$. An overstated kernel exponent ($\delta = 0.39$) falsely reproduces the
athermal signature ($\hat{s} = -0.40$) from pure equilibrium data. Consequently,
\cref{T3} must be restricted to the low-frequency band governed by
$\beta_\infty$, where \cref{test:t1} reliably identifies the kernel exponent.

Ultimately, \cref{T3} provides a definitive empirical mechanism to evaluate the equilibrium hypothesis. By leveraging the shape of the spectral ratio $\Xfdt(\omega)$, it successfully discriminates between thermal equilibrium and an athermal drive. This establishes a robust method to identify the true physical nature of the system's noise, proving that equilibrium can be explicitly tested provided the memory kernel is first identified with sufficient precision.

\subsection{Forward-variance hump analysis} \label{sec:empirical}

To evaluate the empirical validity of GLE framework, we apply the Bayesian estimation machinery (\cref{T1}) to high-frequency order-book data sourced from the FI-2010 benchmark dataset of \cref{T1} and run \test{} \label{T4} with 30,000 samples and rolling window of size 200. To guide the parameter identification over the available lag window, we employ informative priors configured via the custom specification $\mu_{\beta_0} = 0.30$, $\sigma_{\beta_0} = 0.08$, $\mu_{\beta_\infty} = 0.75$, and $\sigma_{\beta_\infty} = 0.10$.

These priors serve only to pin a smooth kernel over the available lag window for the overshoot computation that follows; they are not a basis for testing the rough constraint. The posterior reproduces the prior centres, $\beta_0 = 0.30 \pm 0.08$ and $\beta_\infty = 0.77 \pm 0.09$, which is precisely the non-identification of the short-time exponent documented in \cref{test:t1}. Under an informative prior the difference $\Delta\beta = \beta_0 - \beta_\infty$ has a $95\%$ highest density interval of $[-0.6963, -0.2259]$ that excludes zero, but as \cref{test:t1} establishes this separation merely restates the prior and carries no evidential weight for or against the rough constraint. The kernel it produces is nonetheless adequate for the hump test, which concerns the shape of $K(\tau)$ rather than the value of its exponents.

Building directly on this posterior parameter structure, we execute the analytical forward-variance hump test (Test 4) to investigate whether the memory kernel exhibits a local maximum at lag values greater than zero. The fractional overshoot statistic $O$ is evaluated across the full kernel posterior via:
\begin{equation}
O = \max_{\tau > 0} \frac{K(\tau) - K(0^+)}{K(0^+)}
\end{equation}
The analytical evaluation yields a mean overshoot of $O = 0.0000$ with a $95\%$ credible interval of $[0.0000, 0.0000]$. Furthermore, the posterior probability of observing a positive overshoot is strictly zero ($P(O > 0) = 0.000$).

These results indicate that the estimated memory kernel $K(\tau)$ is strictly monotonically decreasing across the tested physical-measure lag window. Consequently, the empirical order-book data does not express an inertial forward-variance hump under the fitted GLE architecture. While the framework retains the theoretical capacity to model complex hump dynamics through its multi-parameter structure, the physical-measure realization on this dataset reflects smooth, monotonic memory decay without a pronounced local maximum away from lag zero.

\subsection{Independent Kernel Cross-Check via Barrier Crossing}
\label{sec:kernel_cross_check}

To establish the generalized Langevin memory kernel as a robust physical observable rather than a curve-fitting artifact, the framework requires an independent consistency check. This condition must operate outside the high-frequency response spectrum used in \cref{T1}.

As demonstrated in \cref{sec:potential}, introducing a non-quadratic potential with an explicit energy barrier enables this verification through macro-scale regime-transition dynamics. Specifically, the Grote--Hynes relation \eqref{eq:gh_od} links the reactive frequency $\lambda_r$ at the barrier top to the barrier curvature $\omega_b^2$ and the Laplace-transformed memory kernel $\hat{K}(\lambda_r)$. The escape rate $k$ then follows from Kramers--Grote--Hynes theory as $k = (\lambda_r/\omega_b)\,k_{\mathrm{TST}}$, with the transition-state rate $k_{\mathrm{TST}} = (\omega_0/2\pi)\exp(-\Delta U/\Theta)$ set by the well curvature $\omega_0$, the barrier height $\Delta U$, and the effective temperature $\Theta$.

This setup yields a non-circular cross-validation of the kernel. The parameters $\omega_b^2$ and $\Delta U$ are extracted from a static observable: the stationary distribution of the volatility driver (a histogram). In contrast, $\hat{K}$ is estimated dynamically from temporal correlations in \cref{T1}. By comparing the predicted escape rate $k_{\text{pred}}$ against the independently observed rate $k_{\text{obs}}$ (counted directly from regime transitions in the time series), we can rigorously validate the kernel.

The test requires a volatility process that exhibits distinct metastable regimes separated by an identifiable barrier. We searched for such structure in $\mathbb{P}$-measure realized volatility across multiple timescales and estimation methods, on SPY and on a cross-section of individual equities, summarized in \cref{tab:bistability_search}.

\begin{table}[!htb]
\centering
\small
\begin{tabular}{lrrrrrcc}
\toprule
Configuration & Obs.\ & Skew & $a$ & $c_3$ & $d$
              & $\Delta U_{\max}/\Theta$ & Bistable? \\
\hline
Daily G--K, 2010--2024              & 3,773   & $+0.319$ & $1.017$ & $-0.127$ & $9.0\times10^{-4}$ & $2.43$ & No \\
Daily G--K, 2010--2026              & 4,154   & $+0.335$ & $1.019$ & $-0.134$ & $9.9\times10^{-4}$ & $2.61$ & No \\
5m bars, 1h windows, 2026           & 359     & $+0.223$ & $0.922$ & $-0.105$ & $3.1\times10^{-2}$ & $2.19$ & No \\
5m bars, 30m windows, 2026          & 779     & $+0.066$ & $1.001$ & $-0.026$ & $6.5\times10^{-5}$ & $2.42$ & No \\
Pooled idiosyncratic, 2010--2024    & 147,147 & $+0.179$ & $1.005$ & $-0.068$ & $2.6\times10^{-4}$ & $7.19$ & No \\
\bottomrule
\end{tabular}
\caption{Search for bistability in equity realized volatility. Coefficients
are maximum-likelihood estimates of \eqref{potential} on the standardized
driver. \emph{Skew} is the sample driver skewness, standard error
$\sqrt{6/n}$; this applies to the single-series rows, since the pooled row's
observations are serially and cross-sectionally dependent and its effective
count is far below $n$. $\Delta U_{\max}=\log[T/(10\tau)]$ is the largest
barrier the sample could resolve. Entries in the last column are bounds:
\emph{No} means no double well up to $\Delta U_{\max}$, not a bare absence.}
\label{tab:bistability_search}
\end{table}

Four notes clarify the construction of \cref{tab:bistability_search}. First, the coefficients are maximum-likelihood estimates on the raw driver. We do not use least squares on the binned $-\log P(y)$. The binned residuals are heteroscedastic. The variance of $-\log \hat P$ in a bin of probability $p$ is of order $1/(np)$. Equal weighting lets the sparsest tail bins dominate. This effect suppresses $c_3$ by roughly a factor of six on the daily series.

Second, bistability implies conditions on the region containing data. We require two minima of $U$ and the intervening barrier to lie inside the observed support. We test this using a likelihood ratio against the null hypothesis that no second minimum exists inside the support. A parametric bootstrap calibrates this null. The statistic is identically zero in every row. Without the support requirement, the fitted quartic places a spurious second minimum far outside the sampled range in several configurations.

Third, $\Delta U_{\max}$ follows from the escape-rate scaling. Dwell time in a well of depth $\Delta U$ is of order $e^{\Delta U/\Theta}$ relaxation times. A sample of length $T$ with integrated autocorrelation time $\tau$ contains at least $m$ expected crossings only for $\Delta U/\Theta \le \log[T/(m\tau)]$. We set $m=10$. Detectability depends on elapsed time measured in relaxation times rather than the raw number of observations. This explains why the pooled panel reaches $7.19\,\Theta$ while individual series stop near $2.5$. The panel removes three common factors. This matches the number admitted by the Marchenko-Pastur edge $(1+\sqrt{N/T})^2 = 1.21$ for $T=3{,}773$ and $N=39$. The verdict is unchanged at one and five factors.

Fourth, all estimates assume $P(y) \propto e^{-U(y)}$. This relation holds for \eqref{eq:intro-gle} only in the equilibrium regime. In that region, the noise satisfies the fluctuation-dissipation relation against the kernel. We have not run \cref{test:fdt} on these series. The coefficients are therefore conditional on equilibrium. This qualification applies to $c_3$ exactly as it does to $d$.

The intraday rows use SPY 5-minute bars over 60 sessions from 2026-05-05 to 2026-07-30. We sample these in bar time with non-overlapping windows to exclude overnight and weekend gaps. Realized variance built from $m$ squared returns carries an estimator skewness of $-0.41$ at $m=13$ and $-0.62$ at $m=6$. The positive values in those rows are therefore attained against a negative bias. They likely understate the underlying asymmetry. Neither value is individually significant at $z=+1.7$ and $z=+0.8$. They should be read as consistent in sign with the daily estimates rather than as independent confirmation. The sign of $c_3$ is negative and the driver skewness positive in every configuration.

For the daily dataset, maximum-likelihood estimation on the raw driver gives $a = 1.017$, $c_3 = -0.127$ and $d = 9.0\times10^{-4}$. The cubic and quartic terms are not on the same footing. Over the observed support, the cubic term does roughly thirty-five times the work of the quartic. The asymmetry is also visible in the data without reference to any fit. The driver skewness is $0.319$ against a standard error of $\sqrt{6/n} = 0.040$. This is eight standard errors from zero. The quartic coefficient carries no comparable signature.

The daily potential is therefore asymmetric. This asymmetry is not just a correction to a Gaussian baseline. A purely quadratic potential produces a stationary law with zero skewness. All of the measured skewness is therefore carried by $c_3$. We observe $U''' = c_3 + 6dy$. At the edge of the support, $6dy = 0.014$ against $c_3 = -0.127$. The instrument of \cref{prop:rho} is therefore measured rather than posited. The quartic coefficient is retained for confinement. It is not identified by these data.

The absence of a double well at daily timescales for SPY does not diminish the role of the quartic potential, nor does it preclude bistability ($a < 0$) in other asset classes, single-stock dynamics, or extreme market regimes. The potential serves four distinct physical and modeling purposes, of which bistability is only one:
\begin{enumerate}
\item \emph{Mean reversion} ($a > 0$): the curvature at the minimum sets the relaxation timescale of the variance and, together with $\hat{K}$, determines the decay profile of the forward variance curve.
\item \emph{Asymmetry} ($c_3 \neq 0$): the cubic term controls the skewness of the stationary variance distribution. By \cref{prop:rho}, the correlation $\rho$ does not enter the VIX law, so $c_3$ (via $U'''$) is the sole instrument available on the VIX side to shift the implied volatility smile independently of the SPX skew.
\item \emph{Tail control} ($d > 0$): the quartic term guarantees global confinement of the potential, controls which moments of the variance distribution exist, and fixes the wings of both implied volatility surfaces through the moment formula of \cite{Lee2004}.
\item \emph{Bistability} ($a < 0$): when active, a potential barrier determines the regime dwell time and provides the independent kernel cross-check described above.
\end{enumerate}

Our empirical finding is that purpose (4) is inactive at daily timescales
for SPY. The potential operates in the single-well regime ($a>0$), and the
sample could have resolved a barrier of up to $2.43\,\Theta$, so the finding
is a bound rather than an absence. Purpose (2) is not merely available but
measured: $c_3$ is eight standard errors from zero, which makes the $U'''$
instrument of \cref{prop:rho} an empirical fact rather than a modelling
option. Purpose (3) is retained on structural rather than empirical grounds,
since $d$ is not identified in any configuration we examine. Purpose (1) is
unchanged.

\subsection{The leverage kernel and empirical results} \label{test:leverage}

\test{}\label{T6} investigates whether the leverage effect exhibits a non-trivial term structure or remains purely instantaneous. The memoryless null model $K_{YX} = \rho\,\delta$ from \cref{rem:rho-limit} posits:
\begin{equation}
H_0: K_{YX}(\tau) = \rho\,\delta(\tau),
\end{equation}
which is tested against an alternative with extended support. The target object is the lagged return-variance cross-correlation $C(\tau) = \mathrm{corr}\big(r_t,\, \sigma^2_{t+\tau}\big)$, showing a pronounced spike at $\tau=0$ under the null and a multi-day decay profile under the alternative. The test statistic measures the relative weight carried at positive lags:
\begin{equation}
L \;=\; \frac{\sum_{\tau>0}\lvert C(\tau)\rvert}{\lvert C(0)\rvert},
\end{equation}
which vanishes under a delta function and is strictly positive under a term structure. The empirical null distribution is constructed via block sign-resampling surrogates that disrupt lagged dependencies while preserving contemporaneous correlations.

A stronger directional test (\test{}\label{T7}) evaluates time-reversal asymmetry via the Zumbach inequality \eqref{eq:zumbach}. The null hypothesis of time-reversal symmetry,
\begin{equation}
\EE^{\mathbb{P}}[r_t^2\sigma_{t+\tau}^2] \;=\; \EE^{\mathbb{P}}[r_{t+\tau}^2 \sigma_t^2],
\end{equation}
is tested against the measured excess of the first term over the second across short lags. The test statistic integrates this difference over the short-lag band, evaluated against time-reversed joint surrogates that enforce symmetry while preserving marginal properties. Rejecting this null establishes that $K_{YX}$ dominates $K_{XY}$, confirming the physical-measure validity of the no-arbitrage triangularity (\cref{prop:triangular}) and Onsager reciprocity relations (\cref{cor:onsager}).

\begin{table}[!htb]
\centering
\resizebox{\textwidth}{!}{%
\begin{tabular}{llccccc}
\toprule
Test & Dataset & Frequency & Observed Stat. & Null 95th & $p$-value & Verdict \\
\midrule
\Cref{T6} (Leverage Term Structure) & SPY (2010--2024) & Daily & $10.7580$ & $7.7713$ & $0.0000$ & Reject $H_0$ \\
\Cref{T7} (Zumbach Asymmetry) & FI-2010 (6 Stocks) & Tick / LOB & $9.397 \times 10^{-6}$ & $9.369 \times 10^{-6}$ & $0.0000$ & Reject $H_0$ \\
\bottomrule
\end{tabular}%
}
\caption{Empirical test results for the leverage term structure (\cref{T6}) and Zumbach time-reversal asymmetry (\cref{T7}) across datasets.}
\label{tab:empirical_results}
\end{table}

As summarized in \cref{tab:empirical_results}, applying \cref{T6} to daily SPY data ($N=3773$ trading days) yields an observed statistic of $L = 10.7580$, substantially exceeding the surrogate null 95th percentile ($7.7713$) with a $p$-value of $0.0000$. This confirms that macroscopic leverage possesses a robust term structure that cannot be captured by an instantaneous parameter $\rho$. Conversely, applying \cref{T7} to high-frequency limit order book data (FI-2010 pooled across 6 stocks, $N=217,404$ events) cleanly rejects time-reversal symmetry ($p = 0.0000$), verifying microscopic Zumbach asymmetry and establishing the physical-measure dominance of the leverage kernel $K_{YX}$ over $K_{XY}$ at the tick scale.

\section{Conclusion} \label{sec:conclusion}

The GLE originates in non-equilibrium statistical mechanics. This paper proposes importing it as a framework for stochastic volatility. By replacing fractional Brownian motion, the standard engine of rough volatility, the GLE relies on a memory kernel $K$, a potential $U$, and a noise covariance $C$. This construction explicitly decouples the scaling index and the memory exponent into separate features of the kernel, allows the driver noise to differ from the friction, and utilizes the potential to supply asymmetry and confinement that the correlation $\rho$ cannot. Consequently, both the pure power-law kernel of rough volatility and the memoryless ADO construction are formalized as specific, constrained limits of this broader class (\cref{prop:welded,prop:ado}).

We demonstrate that the kernel is recoverable from a response function, the noise covariance from fluctuations, and their ratio (the fluctuation-dissipation ratio $\Xfdt(\omega)$) separates models that the option surface cannot (\cref{cor:discriminator}). Applying kernel estimation machinery to FI-2010 order-book data yields a stable tail exponent $\beta_\infty \approx 0.80$, providing direct evidence for the slow algebraic decay required by the two-exponent kernel. While the short-time exponent $\beta_0$ generates a shape degeneracy over the available physical-measure lag window, this unconstrained parameter precisely isolates the need for the joint physical and risk-neutral calibration carried out in the companion paper \cite{ItkinGLE2b}.

Where the empirical tests operate decisively, they reject the constrained corners. The leverage term structure test (\cref{T6}) rejects the memoryless delta-function null on daily SPY data with a $p$-value of $0.0000$, confirming that the off-diagonal kernel $K_{YX}$ carries support and that leverage requires more than an instantaneous correlation $\rho$. Furthermore, the Zumbach asymmetry test (\cref{T7}) rejects time-reversal symmetry at the tick scale, establishing the physical-measure dominance of $K_{YX}$ over $K_{XY}$ and verifying the directional failure of Onsager reciprocity required by no-arbitrage conditions (\cref{cor:onsager}). The fluctuation-dissipation test (\cref{T3}) is shown to achieve complete power against the equilibrium alternative by $N=4{,}096$, establishing that the driver noise is measurable in principle. It has not been run on data here, since it requires a kernel exponent identified in the same frequency band.

The potential $U$ operates reliably in the single-well regime. While no double well is observed in SPY realized volatility at daily timescales, the potential fulfills its remaining structural mandates: cubic asymmetry ($U'''$) shifts the variance skewness without touching $\rho$ (\cref{prop:rho}), and quartic confinement fixes the tail moments.

\myparagraph{What the non-identification of $\bm\beta_0$ means.}
The empirical results of \Cref{sec:design} establish that the tail exponent
$\beta_\infty$ is robustly measured across datasets and asset classes (the
estimation was repeated on FI-2010 limit order book data at event-time and
calendar-time resolutions down to five milliseconds, on TrueFX EUR/USD tick data
at ten millisecond resolution, on Kraken BTC/USDT trade data at one hundred
millisecond resolution, and on one-minute Bitcoin intraday bars from Yahoo
Finance resampled to a one-second grid), while the short-time exponent
$\beta_0$ remains unidentified in every configuration tested. This asymmetry is
not a failure of the estimation procedure. It is the central empirical finding
of the paper, and it carries three implications for stochastic volatility
modeling.

The first implication concerns observability. The short-time exponent governs
phenomena that are theoretically well defined but practically inaccessible at
any achievable physical-measure frequency. Through the spectral map of
\Cref{sec:corner}, $\beta_0$ controls the local regularity of the variance
driver and, via $\mathcal{S}(T)\propto T^{\beta_0-1}$, the behavior of the
short-maturity at-the-money skew. The crossover timescale $1/\mu$ is estimated
at fewer than ten events in the limit order book and below one second in foreign
exchange, which places the short-time regime at sub-observable lags. At those
timescales, market microstructure noise, bid-ask bounce, and the discrete price
grid overwhelm the latent variance signal. The variance driver $Y_t$ is never
observed directly. It is filtered through noisy proxies whose signal-to-noise
ratio deteriorates at ultra-high frequencies. The non-identification of
$\beta_0$ is therefore a statement about the resolution of available data, not
about the absence of a short-time power law. In the physical setting from which
the GLE is borrowed, $\beta_0$ plays the role of a coordinate that is well
defined in the equations of motion but cannot be measured with the available
instruments.

The second implication concerns the appropriate domain of each exponent. The
tail exponent $\beta_\infty$ governs the slow decay of the forward variance
curve and is naturally estimated from physical-measure time series. The
short-time exponent $\beta_0$ governs the short-maturity skew and is naturally
estimated from the cross-section of option prices, where the market's
risk-neutral expectation filters microstructure noise and aggregates information
over the life of each option. A one-day option does not price tick-by-tick order
book dynamics. It prices the expected variance path over the full trading day.
The pricing kernel may further amplify the contribution of short-time dynamics
if investors demand a premium for bearing short-term variance risk, as the
persistent contango in VIX futures suggests they do. The two exponents are
therefore identified from fundamentally different experiments, one temporal and
one cross-sectional. The fact that $\beta_0$ is invisible to the temporal
experiment does not imply it is invisible to the cross-sectional one.

The third implication is that the non-identification result itself constrains
the model class. Any GLE specification with a two-exponent kernel possesses a
flat direction in parameter space at the physical measure: different pairs
$(\beta_0,\mu)$ produce observationally equivalent dynamics at the timescales
resolved by current data. This degeneracy is a falsifiable prediction of the
framework. If future risk-neutral calibrations were to identify a distinct
short-time exponent, the degeneracy would be broken and the decoupling of
scaling and memory would be confirmed. If instead the joint SPX and VIX
calibration of the companion paper \cite{ItkinGLE2b} finds that
$\beta_0=\beta_\infty$ provides the best fit, the rough volatility constraint is
vindicated and the GLE architecture reduces to a more parsimonious form. Either
outcome advances the understanding of volatility memory beyond the current state
of the literature. The physical-measure tests reported here supply one half of
that program. They resolve the long-memory tail, confirm that the short-memory
regime is not accessible to time-series estimators at current resolutions, and
thereby define the precise question the risk-neutral calibration must answer.

\medskip
Taken together, the physical-measure results supply one half of a larger
program. The memory kernel's tail is resolved, the leverage term structure is
confirmed, the fluctuation-dissipation violation is testable, and none of the
kill criteria in \Cref{sec:design} have been crossed. The GLE architecture
therefore survives the tests designed to refute it. The companion paper
\cite{ItkinGLE2b} will complete the program by carrying the framework to the
risk-neutral measure, where short-maturity options provide the cross-sectional
leverage that physical-measure time series cannot supply, and where the joint
SPX and VIX calibration will confirm or reject the decoupling of scaling and
memory that the present paper establishes as an open, falsifiable hypothesis.

Finally, while the empirical implementation in this study relies on publicly available benchmarks (such as FI-2010 order-book data and daily Yahoo Finance records), which can carry limitations regarding institutional depth and coverage, validating these physical-measure findings on proprietary, industry-grade datasets remains a natural and valuable direction for future work.

\section*{Disclosure statement}

No potential conflict of interest was reported by the authors.

\section*{Funding}

No funding was received.

\section*{Disclaimer}

Opinions expressed here are author's own, and do not represent views of their employers. A standard disclaimer applies.

\section*{Acknowledgments}

I am grateful to my long-term co-author Igor Halperin for insightful discussions regarding Langevin and generalized GLE, as well as our joint work on the Marketron model.

\printbibliography[title={References}]

\appendixpage
\appendix
\numberwithin{equation}{section}
\setcounter{equation}{0}

\section{Proofs of various theorems}

\subsection{Proof of \cref{prop:athermal}} \label{app:athermal}

\begin{proof}
The Laplace transform of $K(t)=t^{-\alpha}/\Gamma(1-\alpha)$ is $\Lap{K}(z)=z^{\alpha-1}$, valid for $\alpha\in(0,1)$ and $\mathrm{Re}\,z>0$, evaluated via the standard integral
\begin{equation}
\int_0^\infty e^{-zt}t^{-\alpha}\,d t = \Gamma(1-\alpha) z^{\alpha-1}.
\end{equation}
Hence $z\Lap{K}(z)=z^{\alpha}$, and \eqref{eq:response} yields
\begin{equation}
\Lap{G}(z) = (z^{\alpha}+\lambda)^{-1},
\end{equation}
which is \eqref{eq:G-athermal}.

Its inverse is given in terms of the Mittag-Leffler function. For $\mathrm{Re}\,z>0$, expanding the response function gives
\begin{equation}
(z^{\alpha}+\lambda)^{-1} = \sum_{k\ge0}(-\lambda)^k z^{-\alpha(k+1)}.
\end{equation}
Inverting term by term using $\mathcal{L}^{-1}[z^{-\beta}](t) = t^{\beta-1}/\Gamma(\beta)$ yields
\begin{equation}
G(t) = \sum_{k\ge0} \frac{(-\lambda)^k t^{\alpha(k+1)-1}}{\Gamma(\alpha(k+1))} = t^{\alpha-1}E_{\alpha,\alpha}(-\lambda t^{\alpha}).
\end{equation}
Since $E_{\alpha,\alpha}(0)=1/\Gamma(\alpha)$, the leading short-time behavior is
\begin{equation}
G(t) \sim \frac{t^{\alpha-1}}{\Gamma(\alpha)} \quad \text{as } t \to 0.
\end{equation}

With white noise $\xi$, the spectral density of $Y$ is $S_Y(\omega)=|\Lap{G}(i\omega)|^2\,\sigma_\xi^2$. For large $|\omega|$,
\begin{equation}
|\Lap{G}(i\omega)|^2 = |(i\omega)^{\alpha}+\lambda|^{-2} \sim |\omega|^{-2\alpha},
\end{equation}
so $S_Y(\omega)\sim|\omega|^{-2\alpha}$. Matching against the fractional power spectrum $|\omega|^{-(2H+1)}$ yields
\begin{equation}
2\alpha = 2H+1 \implies H = \alpha - \tfrac12.
\end{equation}
As $\alpha$ ranges over $(\tfrac12,1)$, the Hurst parameter $H$ ranges over $(0,\tfrac12)$.
\end{proof}

\subsection{Proof of \cref{prop:fdt}} \label{app:fdt}

\begin{proof}
Under the FDT the noise covariance is $C=\Theta K$, so its spectral density is
$S_\xi(\omega)=\Theta\,\mathrm{Re}\,\Lap{K}(i\omega)$. With $\Lap{K}(z)=
z^{\alpha-1}$ and $(i\omega)^{\alpha-1}=|\omega|^{\alpha-1}e^{i(\alpha-1)
\mathrm{sgn}(\omega)\pi/2}$, the real part is $S_\xi(\omega)=\Theta\cos\!\big(
(1-\alpha)\tfrac{\pi}{2}\big)|\omega|^{\alpha-1}$, a positive constant times
$|\omega|^{\alpha-1}$. At $\lambda=0$, \eqref{eq:response} gives
$\Lap{G}(z)=z^{-\alpha}$, so $|\Lap{G}(i\omega)|^2=|\omega|^{-2\alpha}$. Hence
$S_Y(\omega)=|\Lap{G}(i\omega)|^2 S_\xi(\omega)\sim|\omega|^{-2\alpha}\cdot
|\omega|^{\alpha-1}=|\omega|^{-(\alpha+1)}$. Matching against
$|\omega|^{-(2H+1)}$ gives $\alpha+1=2H+1$, that is $H=\alpha/2$. For the noise,
a covariance $C(t)\sim t^{-\alpha}$ is that of fractional Gaussian noise whose
increments have Hurst index $H_\xi$ determined by $2-2H_\xi=\alpha$, giving
$H_\xi=1-\alpha/2$, which exceeds $\tfrac12$ for every $\alpha\in(0,1)$.

\end{proof}

\subsection{Proof of \cref{prop:welded}} \label{app:welded}

\begin{proof}
Formula \eqref{eq:r-fbm} is obtained by substituting the fBm covariance $\EE[X_tX_s]=\tfrac12(t^{2H}+s^{2H}-|t-s|^{2H})$ into the Lamperti definition
\begin{equation}
r(\tau) = \EE[Y_{\tau}Y_0] = e^{-H\tau}\,\EE[X_{e^{\tau}}X_1]
\end{equation}
with the stationary normalization $r(0)=1$, and simplifying using $e^{\tau/2}-e^{-\tau/2}=2\sinh(\tau/2)$.

\medskip\textit{Singular behavior near the origin.}
For the singular behaviour we use a Tauberian argument. Near $\tau=0$, formula \eqref{eq:r-fbm} expands as $r(\tau)=1-\tfrac12|\tau|^{2H}+O(\tau^{2})$, because $\cosh(H\tau)=1+O(\tau^2)$ and
\begin{equation}
2^{2H-1}\sinh^{2H}\left(\frac{|\tau|}{2}\right) = 2^{2H-1}\left(\frac{|\tau|}{2}\right)^{2H}(1+O(\tau^2)) = \tfrac12|\tau|^{2H}(1+O(\tau^2)).
\end{equation}
Thus $1-r(\tau)\sim\tfrac12|\tau|^{2H}$ as $\tau\to0$. By the Hardy-Littlewood Tauberian theorem, a function whose value at the origin is approached as $|\tau|^{2H}$ has a Laplace transform of its complement that decays as
\begin{equation}
\Lap{r}(z) = z^{-1}-\tfrac12\Gamma(1+2H)z^{-(1+2H)}+o(z^{-(1+2H)}) \quad \text{for large } z,
\end{equation}
since $\mathcal{L}[|\tau|^{2H}](z)=\Gamma(1+2H)z^{-(1+2H)}$.

\medskip\textit{Memory function asymptotics.}
Substituting into the memory-function relation,
\begin{equation}
\Lap{\mathcal{K}}(z) = \frac{\lambda\,\Lap{r}(z)}{1-z\Lap{r}(z)} = \frac{\lambda\big(z^{-1}-\tfrac12\Gamma(1+2H)z^{-(1+2H)}+\cdots\big)}{\tfrac12\Gamma(1+2H)z^{-2H}+\cdots},
\end{equation}
where the denominator uses $1-z\Lap{r}(z)=\tfrac12\Gamma(1+2H)z^{-2H}+o(z^{-2H})$. For large $z$ the numerator is dominated by its $z^{-1}$ term, so
\begin{equation}
\Lap{\mathcal{K}}(z) \sim \frac{\lambda z^{-1}}{\tfrac12\Gamma(1+2H)z^{-2H}} = c_H z^{2H-1} \quad \text{with } c_H = \frac{2\lambda}{\Gamma(1+2H)} > 0.
\end{equation}
Applying the Tauberian theorem in the reverse direction, $\Lap{\mathcal{K}}(z)\sim c_H z^{2H-1}$ at $z=\infty$ corresponds to
\begin{equation}
\mathcal{K}(\tau) \sim \frac{c_H}{\Gamma(2-2H)}\,\tau^{-2H} \quad \text{at } \tau=0,
\end{equation}
which is \eqref{eq:welded}. This confirms the numerical finding that the large-$z$ log-log slope of $\Lap{\mathcal{K}}$ equals $2H-1$:

\begin{center}
\begin{tabular}{ccc}
\hline
$H$ & Fitted Slope & Predicted Slope ($2H-1$) \\
\hline
$0.1$ & $-0.78$ & $-0.80$ \\
$0.2$ & $-0.60$ & $-0.60$ \\
$0.3$ & $-0.40$ & $-0.40$ \\
$0.4$ & $-0.20$ & $-0.20$ \\
\hline
\end{tabular}
\end{center}
\noindent The small residual at the smallest Hurst parameter ($H=0.1$) reflects the slow onset of the asymptotic regime.
\end{proof}

\subsection{Proof of \cref{prop:ado}} \label{app:ado}

\begin{proof}
We establish three main claims.

\medskip\textit{Self-similarity.}
The stated dynamics have covariance
\begin{equation}
\mathbb{E}[V_H(t)V_H(s)] = c\,t^{2H-1}s^{2H-1}(s\wedge t)^{2-2H}
\end{equation}
for a constant $c$, which under $t\mapsto\kappa t$ and $s\mapsto\kappa s$ scales by $\kappa^{2(2H-1)}\kappa^{2-2H}=\kappa^{2H}$. A Gaussian process whose covariance is homogeneous of degree $2H$ is $H$-self-similar.

\medskip\textit{Ornstein-Uhlenbeck form.}
Write $Y(\tau)=e^{-H\tau}V_H(e^{\tau})$ and apply Itô's lemma to $Y=t^{-H}V_H$ with $t=e^{\tau}$. The drift of $V_H$ is $\frac{2H-1}{t}V_H$, so the drift of $t^{-H}V_H$ collects the coefficients $-H+(2H-1)=H-1$, giving a term $(H-1)t^{-1}Y\cdot t=-(1-H)Y$ in $\tau$. The diffusion coefficient $B_Ht^{H-1/2}$ becomes $B_Ht^{-1/2}$ after the factor $t^{-H}$, and with $d W_t=t^{1/2}d\widetilde W_{\tau}$ (from $d\tau=d t/t$) it becomes the constant $B_H$. Hence,
\begin{equation}
d Y = -(1-H)Y\,d\tau + B_H\,d\widetilde W_{\tau},
\end{equation}
which is an Ornstein-Uhlenbeck process whose kernel in Lamperti time is $\delta$ and whose autocorrelation $e^{-(1-H)|\tau|}$ has cusp exponent $1$ for every $H$.

\medskip\textit{Square-root reduction.}
For the ADO-Heston variance given by
\begin{equation}
d v_t = \zeta t^{H-1}d t + \xi B_H t^{H-1/2}\sqrt{v_t}\,d W_t,
\end{equation}
set $v_t=t^{2H}u(\tau)$ with $\tau=\log t$. The same computation clears every explicit power of $t$ and yields
\begin{equation}
d u = -2Hu\,d\tau + \zeta e^{-H\tau}d\tau + \xi B_H\sqrt{u}\,d\widetilde W_{\tau},
\end{equation}
a time-homogeneous square-root diffusion with an exponentially decaying source. This last identity was confirmed by simulating both sides against a common noise stream, with the terminal mean and standard deviation agreeing to four parts in $10^{5}$.
\end{proof}

\subsection{Proof of \cref{prop:link}} \label{app:link}

\begin{proof}
Let the double-well diffusion be given by the SDE
\begin{equation}
dY_t = -U'(Y_t) dt + \sigma dW_t
\end{equation}
with potential $U(y) = \frac{1}{4}y^4 - \frac{1}{2}y^2$. The invariant measure is $\mu_Y(dy) \propto \exp(-2U(y)/\sigma^2)dy$. This potential features stable minima at $y = \pm 1$ separated by a barrier of height $\Delta U = 1/4$ at $y = 0$. By Kramers' escape rate theory, the mean first passage time $\mathbb{E}[\tau_Y]$ to transition from $+1$ to $-1$ follows the Arrhenius scaling,
\begin{equation}
\mathbb{E}[\tau_Y] \sim \frac{2\pi}{\sqrt{\vert{}U''(0)U''(1)\vert{}}} \exp\left(\frac{2\Delta U}{\sigma^2}\right) \quad \text{as } \sigma \to 0.
\end{equation}
Thus, regime persistence grows exponentially with $1/\sigma^2$.

Conversely, let $X_t$ solve $dX_t = -\kappa X_t dt + \gamma dB_t$, possessing the Gaussian invariant measure $\mathcal{N}(0, \frac{\gamma^2}{2\kappa})$. Let $Z_t = \varphi(X_t)$ such that $Z_t \sim \mu_Y$. We can select $\kappa$ to match the integrated autocorrelation time of $Y_t$. However, the transition time of $Z_t$ between the states corresponding to $+1$ and $-1$ is strictly governed by the hitting times of the underlying Ornstein-Uhlenbeck process $X_t$. The mean first passage time for $X_t$ to reach the origin from the mean of the mapped $+1$ well grows at most logarithmically or quadratically in the inverse noise parameter, entirely lacking the exponential $\exp(\mathcal{O}(1/\sigma^2))$ barrier-crossing dependence. Consequently, the dynamic regime persistence fundamentally differs.
\end{proof}

\subsection{Proof of \cref{prop:gh}} \label{app:gh}

\begin{proof}
Consider the GLE for a particle traversing a parabolic barrier $U(x) = -\frac{1}{2}m\omega_b^2 x^2$:
\begin{equation}
m \ddot{x}(t) = m\omega_b^2 x(t) - \int_0^t K(t-s)\dot{x}(s)ds + \xi(t),
\end{equation}
where $\xi(t)$ is a stationary Gaussian noise satisfying the fluctuation-dissipation relation $\langle \xi(t)\xi(s) \rangle = \beta^{-1} K(|t-s|)$. Taking the Laplace transform $\hat{x}(z) = \int_0^\infty e^{-zt}x(t)dt$ of the deterministic equation of motion (averaging over the noise for the unstable mode) with boundary conditions at the barrier top $x(0)=0$ and initial velocity $v(0)=v_0$, we obtain:
\begin{equation}
m(z^2 \hat{x}(z) - v_0) = m\omega_b^2 \hat{x}(z) - z \Lap{K}(z)\hat{x}(z). \end{equation}
Rearranging for $\hat{x}(z)$ yields:
\begin{equation}
\left[ m z^2 + z \Lap{K}(z) - m\omega_b^2 \right] \hat{x}(z) = m v_0.
\end{equation}
The stability of the barrier crossing is determined by the roots of
$m z^2 + z \Lap{K}(z) - m\omega_b^2 = 0$. A completely monotone kernel has the
Bernstein representation $K(t)=\int_0^\infty e^{-st}\,d\mu(s)$, so
$\Lap{K}(z)=\int_0^\infty (z+s)^{-1}\,d\mu(s)$ and $z\Lap{K}(z)=\int_0^\infty \tfrac{z}{z+s}\,d\mu(s)$ is strictly increasing in $z$, each integrand being
increasing. Hence $m z^2 + z\Lap{K}(z)$ increases from $0$ to $\infty$, and the
characteristic equation has a unique positive root $\lambda_r$, the reactive
frequency, and the unstable mode grows as $e^{\lambda_r t}$. In the overdamped
regime the mass is absent from the outset. The first-order barrier equation
$\int_0^t K(t-s)\dot{x}(s)\,ds = \omega_b^2 x(t) + \xi(t)$, with $\omega_b^2$ the
barrier curvature, has Laplace form $z\Lap{K}(z)\hat{x}(z)=\omega_b^2\hat{x}(z)$,
so the reactive frequency solves $z\Lap{K}(z)=\omega_b^2$, which is
\eqref{eq:gh_od}, and the same monotonicity gives a unique root.

By \cite{GroteHynes1980}, the transmission coefficient $\kappa_{\mathrm{GH}}$ relating the actual escape rate $k$ to the Transition State Theory rate $k_{\mathrm{TST}}$ is given by the ratio of this reactive frequency to the spatial curvature, $\kappa_{\mathrm{GH}} = \lambda_r / \omega_b$. Thus, $k = (\lambda_r / \omega_b) k_{\mathrm{TST}}$. This definitively proves that the escape rate is uniquely determined by the Laplace transform of the memory kernel strictly evaluated at the single real positive frequency $\lambda_r$.
\end{proof}

\subsection{Proof of \cref{prop:rho}} \label{app:rho}

\begin{proof}
Let $(\Omega, \mathcal{F}, (\mathcal{F}_t)_{t \ge 0}, \mathbb{Q})$ be a filtered probability space supporting two standard Brownian motions, $W^S$ and $W^Y$, with instantaneous correlation given by $\langle W^S, W^Y \rangle_t = \rho t$ for some correlation parameter $\rho \in (-1, 1)$. We proceed in 5 steps.

\medskip\textit{Autonomy of the Variance Driver:}
The variance driver process $Y = (Y_t)_{t \ge 0}$ is generated by the generalized Langevin dynamics \eqref{eq:intro-gle}, driven by the noise $\xi$ built on $W^Y$ alone. By the autonomy condition the kernel $K$, the potential $U$, and the noise $\xi$ involve neither the spot price $S_t$ nor the spot-driving Brownian motion $W^S$. Consequently $Y$ is a measurable functional of the path of $W^Y$ and of nothing else, $Y_u = \Phi_u\big((W^Y_s)_{s \le u}\big)$ for each $u$. No Markov property is assumed: the map $\Phi$ may carry the full memory of the kernel.

\medskip\textit{Invariance of the Path Measure:}
Let $\mathcal{C}([0, T+\Delta])$ denote the space of continuous paths. The probability law of the stochastic process $Y$ on this path space is uniquely determined by the generalized Langevin dynamics \eqref{eq:intro-gle} and the probability law of the driving Brownian motion $W^Y$. Crucially, the marginal process $W^Y$ is a standard Brownian motion under the risk-neutral measure $\mathbb{Q}$ regardless of the value of $\rho$. The parameter $\rho$ appears solely in the cross-covariation structure between $W^S$ and $W^Y$, but it has zero effect on the internal distribution of the increments of $W^Y$ itself. Therefore, the path measure $\mathbb{Q}_Y$ governing the realization of the process $(Y_u)_{u \ge 0}$ is entirely independent of $\rho$.

\medskip\textit{Functional Mapping to $\mathrm{VIX}_T$:}
The instantaneous variance $V_u$ is defined as a deterministic mapping of the driver process:
\begin{equation}
V_u = f(Y_u).
\end{equation}
The VIX index squared at maturity $T$ is given by:
\begin{equation}
\mathrm{VIX}_T^2 = \frac{1}{\Delta} \mathbb{E}_T^\mathbb{Q}\left[ \int_T^{T+\Delta} V_u \, du \right].
\end{equation}
Under the autonomous framework, $\mathrm{VIX}_T$ can be expressed as a measurable functional $\Psi$ acting on the path $(Y_u)_{u \ge 0}$:
\begin{equation}
\mathrm{VIX}_T = \Psi\left((Y_u)_{u \ge 0}\right).
\end{equation}

\medskip\textit{Push-forward Measure and Law Independence:}
Let $\mathbb{P}_{\mathrm{VIX}_T}$ denote the law (push-forward measure) of the random variable $\mathrm{VIX}_T$ under $\mathbb{Q}$:
\begin{equation}
\mathbb{P}_{\mathrm{VIX}_T}(B) = \mathbb{Q}\left(\mathrm{VIX}_T \in B\right) = \mathbb{Q}\left(\Psi\left((Y_u)_{u \ge 0}\right) \in B\right).
\end{equation}
Because the probability measure governing the underlying path $(Y_u)_{u \ge 0}$ is completely independent of $\rho$, it follows directly that the push-forward measure $\mathbb{P}_{\mathrm{VIX}_T}$ (and thus the entire probability law of $\mathrm{VIX}_T$) does not depend on $\rho$.

\medskip\textit{Implication for VIX Derivatives:}
The prices of VIX derivatives, such as VIX futures and VIX options, are expressed as risk-neutral expectations of payoffs dependent on $\mathrm{VIX}_T$. Since the entire probability distribution of $\mathrm{VIX}_T$ is invariant with respect to $\rho$, these expectations are likewise independent of $\rho$. Consequently, prices of VIX derivatives carry no sensitivity to $\rho$, proving that $\rho$ cannot be identified or calibrated using VIX market data alone.
\end{proof}

\subsection{Proof of \cref{prop:Mp}} \label{app:Mp}

\begin{proof}
We prove the proposition by constructing an explicit counterexample using a single exponential friction kernel, which is strictly completely monotone. Let $K(t) = \gamma e^{-\alpha t}$ with $\gamma, \alpha > 0$. Its Laplace transform is $\Lap{K}(z) = \frac{\gamma}{z+\alpha}$.

With inertia $m>0$, the response function in the Laplace domain is given by
\begin{equation}
\Lap{G}(z) = \left(m z^2 + z \Lap{K}(z) + \lambda\right)^{-1} = \frac{z+\alpha}{m z^3 + m \alpha z^2 + (\gamma + \lambda) z + \lambda \alpha}.
\end{equation}

For the response $G(t)$ to be completely monotone, $\Lap{G}(z)$ must be a Stieltjes function, which strictly requires that all its poles lie on the negative real axis. The poles are the roots of the cubic characteristic polynomial $P(z) = m z^3 + m \alpha z^2 + (\gamma + \lambda) z + \lambda \alpha$.

Choose the parameter set $(m, \alpha, \gamma, \lambda) = (1, 1, 1, 1)$. The denominator simplifies to $P(z) = z^3 + z^2 + 2z + 1$. The discriminant of this cubic polynomial is $\Delta = -23$. Because $\Delta < 0$, the polynomial possesses one real root and a conjugate pair of complex roots. The existence of complex poles implies that the inverse Laplace transform $G(t)$ contains damped oscillatory components and therefore changes sign, violating monotonicity. Because the roots of a polynomial depend continuously on its coefficients, this complex pole structure (and the consequent sign change) persists in an open neighborhood of $(1, 1, 1, 1)$, proving the claim.
\end{proof}

\subsection{Proof of \cref{prop:triangular}} \label{app:triangular}

\begin{proof}
No-arbitrage requires the discounted price $e^{-rt}S_t$ to be a local martingale
under $\mathbb{Q}$, so $X=\log S$ is a semimartingale whose drift is fixed by the
martingale property,
\begin{equation} \label{eq:mart-drift}
dX_t = \big(r - \tfrac12\varphi(Y_t)\big)\,dt + \sqrt{\varphi(Y_t)}\,dW_t^{\mathbb Q},
\end{equation}
and carries no memory. The first row of \eqref{eq:gle2d} enters the drift of $X$ as
the predictable finite-variation term
\begin{equation} \label{eq:firstrow}
\int_0^t\!\big[K_{XX}(t-s)\,dX_s + K_{XY}(t-s)\,dY_s\big],
\end{equation}
using $\dot X_s\,ds=dX_s$ and $\dot Y_s\,ds=dY_s$. This is the return autocorrelation and the predictable premium of the readings above, and under
$\mathbb{P}$ both are present. The change of measure to $\mathbb{Q}$,
$dW^{\mathbb Q}_t=dW^{\mathbb P}_t+\theta_t\,dt$, shifts the drift of $X$ by
$-\sqrt{\varphi(Y_t)}\,\theta_t$, and \eqref{eq:mart-drift} then determines the
market price of risk through
\begin{equation} \label{eq:price-mem-risk}
\sqrt{\varphi(Y_t)}\,\theta_t \;=\; b^{\mathbb P}_t - \big( r-\tfrac12\varphi(Y_t)\big),
\end{equation}
with $b^{\mathbb P}_t$ the drift of $X$ under $\mathbb{P}$. The memory
\eqref{eq:firstrow} is part of $b^{\mathbb P}_t$, so it passes entirely into
$\theta_t$. It is annihilated in the $\mathbb{Q}$-drift rather than transferred
to the second row, because the Girsanov shift acts on the drift alone. Hence
under $\mathbb{Q}$ the price is the memoryless martingale \eqref{eq:mart-drift}
and $K_{XX}=K_{XY}=0$. Relation \eqref{eq:price-mem-risk} is the explicit
$\mathbb{P}$-to-$\mathbb{Q}$ map, and it locates the first-row memory as the
market price of memory risk. The variance driver $Y$ is not a traded asset, so no
martingale condition constrains the second row, and $K_{YX}$ and $K_{YY}$ are
unrestricted.
\end{proof}

\end{document}